\documentclass{ws-ijmpa}

\usepackage[super]{cite}
\usepackage{xcolor}
\usepackage[verbose,hypertexnames=false]{hyperref}
\hypersetup{colorlinks=false,allbordercolors=blue,pdfborderstyle={/S/U/W 1}}
\usepackage[vcentermath]{youngtab}

\usepackage{young}
\usepackage{slashed}

\newcommand{\be}{\begin{equation}}
\newcommand{\ee}{\end{equation}}
\newcommand{\bea}{\begin{eqnarray}}
\newcommand{\eea}{\end{eqnarray}}

\begin{document}

\markboth{Ling-Xiao Xu}{Proving $\chi$SB in QCD-like Theories from Anomaly Matching}

%
\catchline{}{}{}{}{}
%

\title{To Break or Not to Break: A Review of a No-Go Theorem on Chiral Symmetry Breaking in QCD-like Theories}

\author{Ling-Xiao Xu}

\address{Abdus Salam International Centre for Theoretical Physics,\\
Strada Costiera 11, 34151, Trieste, Italy\\
lxu@ictp.it}

\maketitle


\begin{abstract}
This is a pedagogical review of some recent progress in rigorously proving chiral symmetry breaking in a class of QCD-like theories that closely resemble the real-world QCD --- the $SU(N_c)$ Yang-Mills theory coupled to $N_f$ flavors of massless quarks in the fundamental representation. Based on ’t Hooft anomaly matching and persistent mass conditions, a general no-go theorem is formulated: assuming that the theory flows in the infrared to a fully color-screened, infrared-free phase described by color-singlet hadrons, symmetry and anomaly constraints necessarily imply spontaneous chiral symmetry breaking; conversely, any phase with unbroken chiral symmetry must retain unscreened color charges, thereby ruling out a fully color-singlet hadron description in the infrared. While these results have been widely assumed, the recent developments reviewed here establish them with a new level of rigor. The persistent mass condition, carefully formulated here, plays a central role --- just as it does in the Vafa–Witten theorem on unbroken vectorlike symmetries.
\end{abstract}


\ccode{PACS numbers: 11.10.-z, 11.15.-q, 11.30.Rd, 11.30.Qc, 12.38.Aw}

\tableofcontents

\section{Introduction and Summary}	
\label{sec:intro}

Our most fundamental understanding of Nature so far is based on the Standard Model (SM) of particle physics, where the theory is constructed with the gauge group $G_{\text{SM}}=SU(3)_c\times SU(2)_L\times U(1)_Y$~\footnote{More precisely, there is an ambiguity regarding the global form of $G_{\text{SM}}$ distinguished by a central $\mathbb{Z}_6$ group which acts trivially on all the SM particles~[\citen{ORaifeartaigh:1986agb, Hucks:1990nw, Tong:2017oea}] (see also~[\citen{Davighi:2019rcd, Wan:2019gqr, Wang:2021ayd, Wang:2020mra}]). From a particle physics perspective, it is pointed out recently that discovering new particles beyond the SM that are charged under the central group can potentially determine the global form of $G_{\text{SM}}$~[\citen{Alonso:2024pmq, Li:2024nuo, Koren:2024xof}].} and the matter fields furnish various chiral representations of $G_{\text{SM}}$. 
At low energies, the electroweak gauge group $SU(2)_L\times U(1)_Y$ is Higgsed~\footnote{See however~[\citen{ Frohlich:1980gj, Frohlich:1981yi, tHooft:1979yoe}] for a gauge-invariant formulation of the Higgs mechanism (see also recently~[\citen{Egger:2017tkd, Maas:2020kda}]). Somewhat contrary to the standard textbook presentation, such a formulation highlights the fact (i.e., the Elitzur theorem) that gauge-dependent order parameters must have vanishing vacuum expectation values~[\citen{Elitzur:1975im}], and gauge ``symmetry'' is only a redundancy in our theoretical description, hence it cannot be spontaneously broken. Relatedly, the lore is that there is no phase transition between the Higgs and confining regions, known as the Higgs-confinement continuity~[\citen{Fradkin:1978dv, Banks:1979fi}].} down to the $U(1)_{\text{EM}}$ electromagnetism, while the color gauge group $SU(3)_c$ gets strongly coupled and all the charged particles --- quarks and gluons --- form into color-singlet hadrons. The theory of the strong interactions based on the $SU(3)_c$ gauge group is known as quantum chromodynamics (QCD), which was established by the discovery of asymptotic freedom~[\citen{Gross:1973id, Politzer:1973fx}].

Although the ultraviolet (UV) theory of QCD is firmly established, rigorously understanding its infrared (IR) phases is surprisingly challenging due to the strongly-coupled nature of the theory.~\footnote{It is well-known that strong coupling versus weak coupling depends on the degrees of freedom being used in the description of the theory; see e.g. as demonstrated by the Seiberg duality~[\citen{Seiberg:1994pq}]. Here, by strong coupling, we mean that QCD is described by quarks and gluons whose interaction is strong at the infrared. (See also e.g.~[\citen{Terning:1997xy, Schmaltz:1998bg, Armoni:2008gg, Sannino:2009qc, Sannino:2009me, Mojaza:2011rw, Sannino:2011mr, Karasik:2022gve, Cacciapaglia:2024mfy}] for various proposals on non-supersymmetric dualities, which may or may not be relevant to the real world.)} 
\begin{itemlist}
\item One particular phenomenon is the so-called confinement, which is only sharply defined in the limit where all the quarks have infinite masses. In modern terminology, this notion of \emph{genuine confinement} is characterized by the electric one-form global symmetry~[\citen{Gaiotto:2014kfa}], which coincides with the center of the gauge group. See~[\citen{Cherman:2022eml, Cherman:2023xok}] on related discussions on the (impossibility of) emergence of one-form symmetry and confinement.
\item On the other hand, in the limit where all the quarks are massless, all color charges are screened at long distances, and along the way the theory exhibits spontaneous chiral symmetry breaking ($\chi$SB)~[\citen{Nambu:1961tp, Nambu:1961fr}].
\end{itemlist}
Some clarification of terminology is needed to avoid confusion: in literature, the phenomenon of ``\emph{color screening}'' is also commonly referred to as ``confinement''. However, strictly speaking, these two notions have to be distinguished~[\citen{Greensite:2011zz, Greensite:2017ajx, Greensite:2018ebg, Greensite:2023qfx}] (see also~[\citen{Dumitrescu:2023hbe}]). Since in the former the Wilson lines are screened, one might use the term ``color screening'' or ``\emph{screening confinement}'' interchangeably.~\footnote{Due to the presence of massless quarks in the fundamental representation, all the Wilson lines are screened by the creation of quark-antiquark pairs. Accordingly, the Wilson lines do not obey the area law. At large distances $r$, the potential between two test charges goes like $V(r)\sim \text{constant}$ instead of $V(r)\sim r$; see e.g.~[\citen{Intriligator:1995au}]. This notion of ``screening confinement'' is also realized in the s-confining supersymmetric gauge theories; see~[\citen{Csaki:1996sm, Csaki:1996zb}] for a systematic classification. On the other hand, the notion of ``genuine confinement'' is well-defined when all quarks fully decouple; this is the limit where the theory reduces to the pure Yang-Mills theory, where the Wilson lines are unscreened and believed to exhibit the area law.}
Here, we are not trying to provide an alternative definition of confinement for QCD with fundamental quarks; rather, we will use the intuitive notion of ``screening confinement'' in the rest of the review, which implies color-singlet hadrons in the IR. 

More broadly, we are interested in a class of QCD-like theories in $(3+1)$ spacetime dimensions. In the UV, the theory is described by the $SU(N_c)$ Yang-Mills theory coupled to $N_f$ flavors of massless quarks in the fundamental representation. 
It is believed that such a theory admits many IR phases depending on the values of $N_c$ and $N_f$: 
\begin{romanlist}[(ii)]
\item \label{phase1} For $N_f\geq \frac{11}{2} N_c$, the theory flows to an IR-free fixed point of quarks and gluons. 
\item \label{phase2} For $N_f^{\text{CFT}}\leq N_f < \frac{11}{2} N_c$, the theory flows to an interacting conformal field theory (CFT). (The precise value of the lower edge $N_f^{\text{CFT}}$ is unknown~\footnote{However, see~[\citen{Neil:2011yag}] for a collection of results from lattice simulations and~[\citen{Gies:2005as, Braun:2006jd, Braun:2009ns, Braun:2010qs, Goertz:2024dnz}] for results from functional renormalization group (fRG) approach. See also e.g.~[\citen{Miransky:1998dh, Miransky:1996pd, Kaplan:2009kr}] for discussion on the dynamics inside the conformal window and conformality lost.}, while the upper edge $11 N_c/2$ is justified by~[\citen{Caswell:1974gg, Banks:1981nn}] in the large $N_c$ limit.)
\item \label{phase3} For $2\leq N_f < N_f^{\text{CFT}}$, the theory flows to an IR-free fixed point of color-singlet hadrons with $\chi$SB.~\footnote{See e.g.~[\citen{Casher:1979vw}] for a nice argument supporting spontaneous $\chi$SB, and~[\citen{Banks:1979yr}] on the relation between $\chi$SB and the spectral property of the Dirac operator. See also~[\citen{Engel:2014cka, Engel:2014eea, Giusti:2015kwf, Faber:2017alm}] for related lattice computations. Furthermore, $\chi$SB is supported by large $N_c$ analysis~[\citen{Coleman:1980mx, Veneziano:1980xs, Sato:2022ayb}].}
\item \label{phase4} For $N_f=1$, the theory is gapped with a unique vacuum. 
\item \label{phase5} For $N_f=0$, the $\theta$ parameter becomes physical in the $SU(N_c)$ Yang-Mills theory, the theory is gapped with a unique vacuum for generic $\theta$, whereas there are two degenerate vacua at $\theta=\pi$. (See~[\citen{Witten:1978bc, Witten:1979vv, Witten:1980sp, DiVecchia:1980yfw,  Witten:1998uka}] for large $N_c$ and~[\citen{Gaiotto:2017yup}] for anomaly arguments supporting such a picture.) 
\end{romanlist}
Notice that such a picture of infrared phases is mainly based on empirical evidence, numerical lattice simulations (for a review, see~[\citen{Neil:2011yag}]), and well-motivated guesswork. Very little has been coherently understood or rigorously proved. For instance, there is the Millennium Prize problem on rigorously showing the existence of the Yang-Mills theory (by taking the continuum limit of a finite lattice system) and a mass gap~[\citen{witten-jaffe}]. Other exotic phases may be possible for special values of $N_c$ and $N_f$.

This needs to be contrasted with the $\mathcal{N}=1$ supersymmetric (SUSY) QCD, where a phase diagram for specific $N_c$ and $N_f$ was coherently understood partially thanks to the power of holomorphy~[\citen{Seiberg:1994bz, Seiberg:1994pq}]; see also Ref.~[\citen{Intriligator:1995au}] for a review. 
By introducing additional SUSY-breaking terms, we might extrapolate from the SUSY limit and develop a coherent understanding of the IR phases of ordinary QCD-like theories~[\citen{Aharony:1995zh, Cheng:1998xg, Arkani-Hamed:1998dti, Luty:1999qc, Abel:2011wv, Murayama:2021xfj, Kondo:2021osz, Luzio:2022ccn, Dine:2022req, Csaki:2022cyg, deLima:2023ebw, Bai:2025dys}]. The hope is that the theory remains under analytic control even near the decoupling limit of all the SUSY partners of quarks and gluons, where in this limit the theory reduces to the ordinary QCD without SUSY. There might be a smooth path of sending the SUSY-breaking scale from zero to infinity; one particular challenge is to prove whether phase transitions occur.~\footnote{For instance, near the SUSY limit, there may be vacua that spontaneously break the baryon number. However, since the baryon number (and other vectorlike symmetries) cannot be spontaneously broken in the non-SUSY limit of QCD~[\citen{Vafa:1983tf}], it is important to understand whether the baryon-number-breaking vacuum is the global minimum in the near-SUSY limit. See e.g.~[\citen{Kondo:2025njf}] for a survey supporting the crossover between the SUSY and non-SUSY limits.}

In this review, we do not attempt to develop a coherent understanding of all the IR phases of QCD-like theories for any specific $N_c$ and $N_f$, or their relations with the corresponding SUSY theories. Instead, by leveraging 't Hooft anomaly matching conditions (AMC)~[\citen{tHooft:1979rat}] and persistent mass conditions (PMC)~[\citen{tHooft:1979rat, Preskill:1981sr, Vafa:1983tf}], we will establish a general no-go theorem following the analysis in Refs.~[\citen{Ciambriello:2022wmh, Ciambriello:2024xzd, Ciambriello:2024msu}]. In particular, we prove rigorously that unbroken chiral symmetry is \emph{incompatible} with a fully color-screened, IR-free phase. More specifically, such a phase is characterized by a renormalization-group (RG) flow, where the UV theory is described by a gauge theory of quarks and gluons, while the IR theory is described by color-singlet hadrons (see Eqs.~\eqref{color_screening} and~\eqref{color_screening_2} for a definition).

Summarizing our main results in~[\citen{Ciambriello:2022wmh, Ciambriello:2024xzd, Ciambriello:2024msu}], we proved the following no-go theorem applicable to theories with generic values of $N_c\geq 3$~\footnote{The case of $N_c=2$ needs to be addressed separately, since when the gauge group is $SU(2)$, quarks and antiquarks both transform as the fundamental representation, which is pseudo-real. Notice that all color-singlet hadrons are bosons in the $SU(2)$ gauge theory with fundamental quarks, suggesting that $\chi$SB necessarily occurs once color screening is assumed. For theories with $N_f$ flavors of Dirac fermions, the breaking pattern is $SU(2N_f)\to Sp(2N_f)$, where the unbroken vectorlike symmetry can be justified using Vafa-Witten like arguments~[\citen{Kosower:1984aw}].}, and any $N_f$ no smaller than the smallest nontrivial prime factor of $N_c$~\footnote{For instance, in the real QCD with $N_c=3$, the theorem applies to all $N_f\geq 3$. We will see in Section~\ref{sec:proof_2} where this constraint on $N_f$ comes from, given a fixed $N_c$.}:
\begin{theorem}\label{thm:main}
In any infrared phase of QCD-like theories with the $SU(N_c)$ gauge group and $N_f$ flavors of massless quarks in the fundamental representation, either the free hadron description with full color screening breaks down, or chiral symmetry $SU(N_f)_L\times SU(N_f)_R$ must be spontaneously broken.
\end{theorem}
Our theorem rules out the existence of the  ``confinement without $\chi$SB'' phase~\footnote{Such a phase is possible for other theories~[\citen{Bars:1981se, Poppitz:2019fnp, Seiberg:1994bz, Csaki:1996sm, Csaki:1996zb}] to which our theorem does not apply. The reason is that, in the presence of Yukawa couplings (as in SUSY theories) or chiral fermions, PMC are not valid as the fermion integration measure in the path integral is not positive definite~[\citen{Vafa:1983tf}]. See e.g.~[\citen{Dimopoulos:1981xc}] for concrete models where PMC fails. See also~[\citen{Valenti:2023olg}] for other perspectives on the positivity of fermion integration measure in the presence of Yukawa couplings. } in QCD-like theories.
Indeed, such a result matches the IR phases~(\ref{phase1}),~(\ref{phase2}),~(\ref{phase3}) discussed earlier where chiral symmetry is well-defined. When the hadronic color-screening description does not apply, as in the phases~(\ref{phase1}) and~(\ref{phase2}), chiral symmetry remains unbroken; otherwise, when the hadronic color-screening description applies, as in the phase~(\ref{phase3}), spontaneous $\chi$SB \emph{necessarily} follows. In the phases~(\ref{phase4}) and~(\ref{phase5}), chiral symmetry is not defined.

We note that, while the results implied by our theorem are widely assumed to be true, the novelty of our work in~[\citen{Ciambriello:2022wmh, Ciambriello:2024xzd}] is to establish these results with a level of mathematical rigor based on robust algebraic and analytic methods in quantum field theory (QFT). These methods are applicable in the strongly coupled regimes of quarks and gluons in QCD-like theories. See also~[\citen{Frishman:1980dq, Farrar:1980sn, Schwimmer:1981yy, Cohen:1981iz, Kaul:1981fd, Takeshita:1981sx}] for early related studies, and~[\citen{Ciambriello:2024msu}] for a survey of all the arguments with concrete examples. 

The remainder of the review is organized as follows. In Section~\ref{sec:anomlies}, we briefly review the 't Hooft AMC and their application in understanding $\chi$SB in QCD-like theories, and then show that $\chi$SB occurs for theories with special values of $N_c$ and $N_f$. In Section~\ref{sec:pmc}, we derive the PMC following Vafa and Witten's seminal work, and then discuss some novel algebraic features of PMC that are crucial for our proof. In Section~\ref{sec:proof_2}, we review the strategy --- called ``\emph{downlifting}'' --- that we use to prove Theorem~\ref {thm:main}. Such a strategy combines the arguments in previous sections. Finally, we conclude in Section~\ref{sec:discussion} while emphasizing the open questions and challenges. Comments on other approaches in the literature can be found in~\ref{app1:other_approaches}.

\section{'t Hooft Anomaly Matching Conditions for Chiral Symmetry}	
\label{sec:anomlies}

Anomalies are among the most fundamental and useful defining data of a QFT. Since anomalies are not renormalized under the RG flow~[\citen{Adler:1969er}], and since the 't Hooft anomaly matching condition (AMC) requires that anomalies in the UV description of a theory must also be matched in the IR~[\citen{tHooft:1979rat}], they provide a powerful and robust tool for constraining emergent phenomena across high energy and condensed matter physics. For reviews on the general foundation of anomalies and their numerous applications, see e.g.~[\citen{Weinberg:1996kr, Alvarez-Gaume:1984zlq, Alvarez-Gaume:1985zzv, Harvey:2005it, Bilal:2008qx, Witten:2015aba, Tachikawa_TASI, Arouca:2022psl, tong_gauge_theory, Cheng:2022sgb, Ye:2023uoz, Cordova:2019jnf, Cordova:2019uob}]. 

One way of characterizing the 't Hooft anomalies is to couple the QFT with a certain global symmetry to a background gauge field $A$. A 't Hooft anomaly is present when the partition function $Z[A]$ fails to be gauge invariant under a background gauge transformation. Under the gauge transformation $A\to A^\lambda$ with $\lambda$ being the gauge parameter, $Z[A]$ transforms by a phase that is a local functional of the gauge parameter $\lambda$ and the gauge field $A$, i.e.,
\be
Z[A^\lambda]=Z[A] \ e^{-2\pi i \int_{X} \alpha(\lambda, A)}\;,
\ee
where the phase cannot be removed by adding local counterterms, and $X$ is the spacetime manifold on which the QFT is defined. Consequently, 't Hooft anomalies imply \emph{obstructions to gauging}. However, 't Hooft anomalies of global symmetries do not imply inconsistency of the original theory.

One can also characterize the anomalies using invertible field theories~[\citen{Freed:2004yc, Freed:2019jzd, Freed:2014iua}] --- a local and classical Lagrangian $-2\pi i w(A)$ --- in one higher spacetime dimension, whose partition function is 
\be
\mathcal{A}[A]= e^{2\pi i \int_Y w(A)}
\ee
where $Y$ is the spacetime manifold whose boundary is $X$ supporting the original dynamical QFT, i.e., $X=\partial Y$. A property of the invertible field theory is that under $A\to A^\lambda$, 
\be
\mathcal{A}[A^\lambda]= \mathcal{A}[A] \ e^{2\pi i \int_{X} \alpha(\lambda, A)}\;,
\ee
such that the total partition function with the auxiliary bulk is gauge invariant under background gauge transformation, i.e., 
\be
Z[A] \ \mathcal{A}[A]=Z[A^\lambda] \ \mathcal{A}[A^\lambda]\;. 
\ee
The 't Hooft anomaly of the original QFT supported on the boundary $X$ inflows from the nontrivial invertible field theory in the bulk $Y$~[\citen{Callan:1984sa}]. In condensed matter physics, the bulk theory in the anomaly inflow is closely related to symmetry-protected topological (SPT) phases~[\citen{Chen:2011pg}]; see e.g.~[\citen{Senthil:2014ooa, Witten:2015aba}] for reviews.

The presence of 't Hooft anomalies in the UV theory, according to the anomaly matching argument, implies that the IR theory cannot be described by a unique symmetry-preserving trivially gapped vacuum. Such a perspective, denoted as \emph{obstruction to a trivial mass gap}, can itself characterize 't Hooft anomalies as a defining property. In condensed matter physics, related results are referred to as the Lieb-Schultz-Mattis (LSM) type theorems~[\citen{Lieb:1961fr, Affleck:1986pq, Oshikawa:2000lrt, Hastings:2003zx}]; see e.g.~[\citen{Affleck:1988nt, Tasaki:2022gka}] for reviews. A third, though less well-known, perspective of characterizing 't Hooft anomalies is that the theory admits no symmetry-preserving boundaries~[\citen{Jensen:2017eof, Thorngren:2020yht}], which can be denoted as \emph{obstruction to symmetry-preserving boundaries}.~\footnote{It is still an open problem to understand the precise relations between different perspectives on how to define 't Hooft anomalies for various global symmetries in QFTs in various spacetime dimensions; see e.g.~[\citen{Li:2022drc, Chen:2023hmm, Choi:2023xjw, Seifnashri:2023dpa}] for related discussions.}

\subsection{Chiral Symmetry in QCD-like Theories}
\label{sec:anomaly_2}

Given the above broad introduction on how to characterize 't Hooft anomalies, we will be interested in revisiting their implication for $\chi$SB in QCD-like theories (including the QCD close to the real-world particle physics, a theory with $N_c=3$ and with massless quarks), a problem originally studied by 't Hooft in his Cargese lectures~[\citen{tHooft:1979rat}].

Let us consider the QCD-like theory with $N_f$ flavors of massless quarks. The theory is invariant under transformations of the chiral symmetry group
\be
\mathcal{G}[N_f]=SU(N_f)_L \times SU(N_f)_R \times U(1)_B\, ,
\label{def:G}
\ee
where the left-handed quarks $q_L$ and the right-handed quarks $q_R$ are charged respectively as 
\be
q_L\sim \left(\scriptsize\yng(1)\ , s\ , \frac{1}{N_c}\right), \quad\quad\quad q_R\sim \left(s\ , \scriptsize\yng(1)\ , \frac{1}{N_c}\right)\; ,
\label{eq:quark_charges}
\ee
where $s$ stands for the singlet representation of either $SU(N_f)_L$ or $SU(N_f)_R$.  
From these quantum numbers, one immediately observes that the perturbative triangle anomalies $[SU(N_f)_{L,R}]^3$ and $[SU(N_f)_{L,R}]^2 U(1)_B$ do not vanish, hence there is obstruction to gauging $\mathcal{G}[N_f]$.~\footnote{There are some discrete group elements within $\mathcal{G}[N_f]$ which act trivially on both $q_L$ and $q_R$~[\citen{Yonekura:2019vyz, Tanizaki:2018wtg, Morikawa:2022liz}]. (See also~[\citen{Ciambriello:2024xzd}] for a careful derivation.) One has to take a quotient to remove these discrete central elements to identify the faithful chiral symmetry group. However, since the perturbative triangle anomalies are not sensitive to the discrete quotient, we neglect it here for simplicity.} To gauge $\mathcal{G}[N_f]$ with background gauge fields, 't Hooft suggests introducing additional spectator fermions to cancel the anomalies induced by the quarks~[\citen{tHooft:1979rat}]. 
The spectator fermions are only charged under the $\mathcal{G}[N_f]$ global chiral symmetry group but neutral under the $SU(N_c)$ gauge group. Using the particle physics terminology, we denote these spectator fermions as the ``leptons'' $L$. In the UV, the anomaly cancellation condition between quarks and leptons is given by
\be
\mathcal{A}(q)+\mathcal{A}(L)=0\;, 
\label{anomaly_cancellation_1}
\ee
where $\mathcal{A}(f)$ is the anomaly coefficient of either $[SU(N_f)_{L,R}]^3$ or $[SU(N_f)_{L,R}]^2 U(1)_B$ for quarks and leptons in the spectrum when $f=q, L$. 
Since the anomalies are canceled, the full theory with both quarks and leptons is invariant under the gauge transformations of the background gauge fields coupled to $\mathcal{G}[N_f]$, and there is no obstruction to gauging $\mathcal{G}[N_f]$.

\begin{figure}[t]
\centerline{\includegraphics[width=5cm]{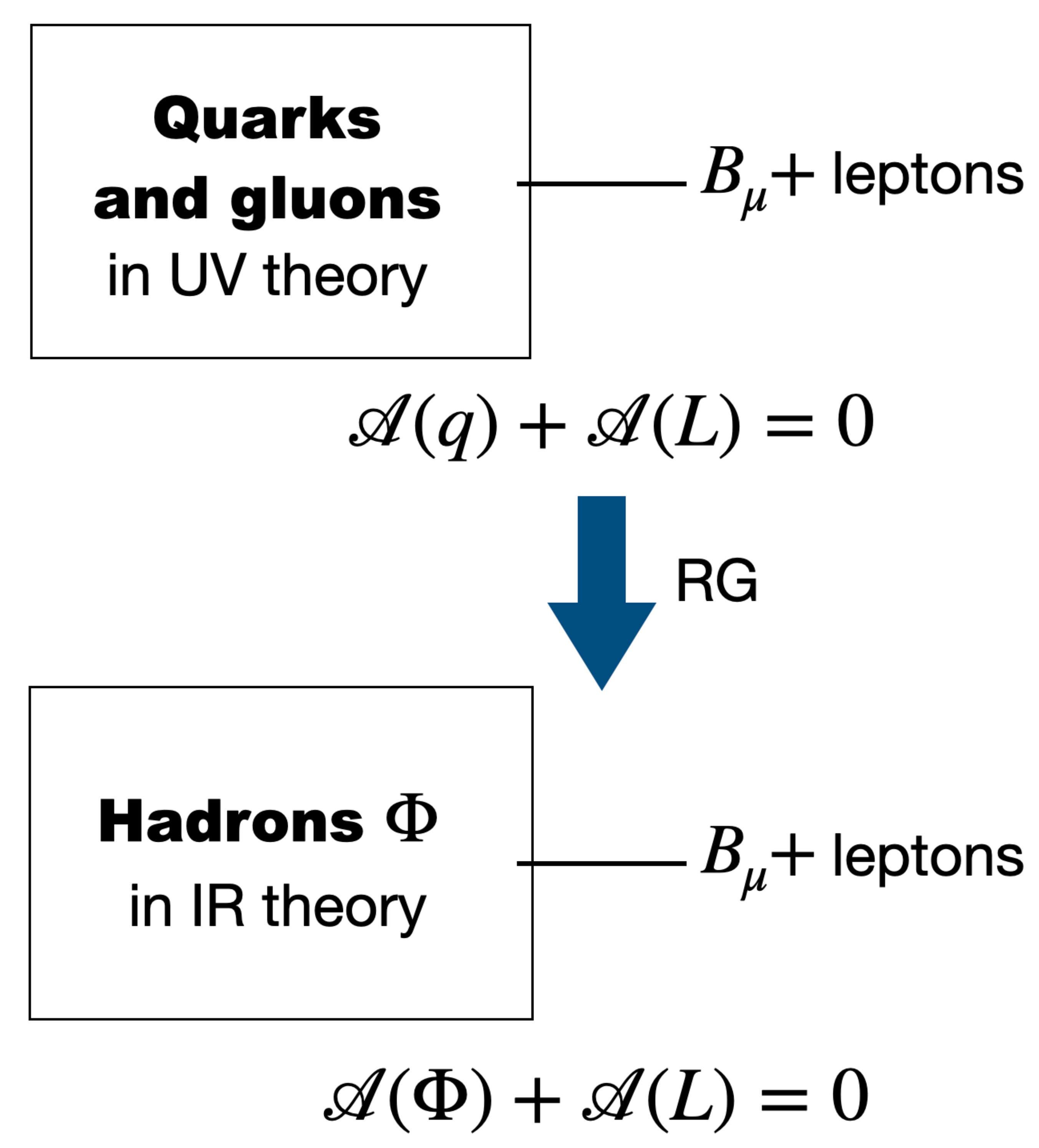}}
\caption{A schematic illustration of the derivation of the 't Hooft AMC, i.e., Eq.~\eqref{AMC_0}. In the UV, the anomalies of quarks are canceled by the color-neutral ``leptons'', allowing the chiral symmetry group $\mathcal{G}[N_f]$ to be gauged by background gauge fields denoted as $B_\mu$. As the theory flows to IR, quarks and gluons confine into color-singlet hadrons while leptons remain in the spectrum. Consequently, the anomalies of the hadrons are canceled by those of the leptons. Comparing the anomaly cancellation conditions in the UV and IR yields Eq.~\eqref{AMC_0}.}
\label{AMC_derivation}
\end{figure}

When the theory flows to an IR-free fixed point with full color screening, there are only hadrons and leptons in the low-energy theory. These particles are all color-singlets. Since the full low-energy theory still has to be gauge invariant under the gauge transformations of the background gauge fields coupled to $\mathcal{G}[N_f]$, the anomalies of hadrons and leptons must cancel each other in the IR, i.e., 
\be
\mathcal{A}(\Phi)+\mathcal{A}(L)=0\;, 
\label{anomaly_cancellation_2}
\ee
where $\mathcal{A}(\Phi)$ denotes the anomaly coefficient of either $[SU(N_f)_{L,R}]^3$ or $[SU(N_f)_{L,R}]^2 U(1)_B$ for all the hadrons $\Phi$ in the spectrum. From Eqs.~\eqref{anomaly_cancellation_1} and~\eqref{anomaly_cancellation_2}, we find
\be
\mathcal{A}(\Phi)=\mathcal{A}(q)\;.
\label{AMC_0}
\ee
This is the 't Hooft AMC~[\citen{tHooft:1979rat}] for the chiral symmetry $\mathcal{G}[N_f]$, where the triangle anomalies of quarks in the UV are matched by those of the hadrons in the IR. See Fig.~\ref{AMC_derivation} illustrating the derivation of Eq.~\eqref{AMC_0}.~\footnote{Furthermore, we refer the reader to~[\citen{Frishman:1980dq, Coleman:1982yg}] for the derivation of AMC from analyticity and unitarity.}

Here is the dynamical implication of Eq.~\eqref{AMC_0}.
\begin{itemlist}
\item When $\chi$SB occurs, the anomalies are matched at IR by the pions through the Wess-Zumino-Witten term~[\citen{Wess:1971yu, Witten:1983tw}]; see also discussions in Refs.~[\citen{Lee:2020ojw, Yonekura:2020upo, Yonekura:2019vyz}] from a modern perspective. 
\item Otherwise, when $\chi$SB does not occur, the anomalies are matched by the massless composite spin-$1/2$ fermions. 
\end{itemlist}
Notice that massless particles that are charged under $G[N_f]$ cannot have spin $>1/2$, as proven by Weinberg and Witten~[\citen{Weinberg:1980kq}]. Furthermore, the assumption of color screening implies that all the hadrons are color singlets, it means that the numbers of the constituent quarks and antiquarks (i.e., $n_q$ and $n_{\bar{q}}$) in each hadron have to satisfy the constraint 
\be
n_q-n_{\bar{q}}=b N_c\;, 
\label{color_screening}
\ee
where $b$ is the baryon number of the hadron being considered. Importantly, $b$ must be an integer for any individual color-singlet hadron, i.e., color screening implies the following condition~\footnote{It is easy to see that any hadron constructed as such must furnish a representation of the $SU(N_c)$ gauge group with $N_c$-ality being zero. Indeed, the $SU(N_c)$ singlet has $N_c$-ality being zero.}
\be
b\in \mathbb{Z}\;.
\label{color_screening_2}
\ee
On the other hand, the hadrons can furnish any representations of $\mathcal{G}[N_f]$ defined in Eq.~\eqref{def:G} as long as the conditions in Eqs.~\eqref{color_screening} and~\eqref{color_screening_2} are satisfied, i.e., we cannot determine whether each hadron in a specific representation forms or not dynamically.~\footnote{See~[\citen{Ciambriello:2022wmh}] for a finer classification of the putative hadrons and e.g.~[\citen{Ciambriello:2024msu}] for numerous examples. Notice that this finer classification is only relevant for the proof of the argument called ``$N_f$ independence''~[\citen{Frishman:1980dq, Farrar:1980sn, Schwimmer:1981yy, Cohen:1981iz, Kaul:1981fd, Takeshita:1981sx}]. Throughout this review, only Eqs.~\eqref{color_screening} and~\eqref{color_screening_2} are relevant.}

More specifically, we are interested in finding solutions for Eq.~\eqref{AMC_0} while assuming ``confinement without $\chi$SB''. Again, the term ``confinement'' refers to the case of screening confinement, as we clarified in Section~\ref{sec:intro}. By evaluating both the anomalies in the UV and IR, Eq.~\eqref{AMC_0} becomes
\begin{equation}
\sum_{r} \,\ell\!\left(r \right) \mathcal{A}\!\left(r\right) = \mathcal{A}\!\left(q\right)\;.
\label{eq:AMC}
\end{equation}
\begin{itemlist}
\item Since both the quark quantum numbers and the spectrum are known, the anomalies in the UV --- the right-hand side of Eq.~\eqref{eq:AMC} --- are fully determined, whose value can always be normalized to a constant which is independent of $N_f$. For instance, let us compute the anomaly coefficient for $q_L$: from the quantum numbers given in Eq.~\eqref{eq:quark_charges}, we find that the anomaly coefficient for the $[SU(N_f)_L]^3$ triangle is $N_c$, while that for the $[SU(N_f)_L]^2 U(1)_B$ triangle is $N_c\cdot \frac{1}{N_c}=1$, where we have normalized the Dynkin index of the fundamental representation of $SU(N_f)_L$ to be one. Similar results apply to $q_R$, where the anomaly coefficients of the $[SU(N_f)_R]^3$ and $[SU(N_f)_R]^2 U(1)_B$ triangles are $N_c$ and $1$, respectively. Notice that the anomaly coefficient of $[U(1)_B]^3$ vanishes, since $q_L$ and $q_R$ have the same charge under $U(1)_B$ whose contributions cancel each other.

More formally, the perturbative anomalies of quarks can be characterized by the six-form anomaly polynomial through the Stora-Zumino descent chain~[\citen{Stora:1983ct, Zumino:1983ew, Manes:1985df}] (see also e.g.~[\citen{Weinberg:1996kr, Cordova:2018cvg, Yonekura:2019vyz, Lee:2020ojw}] for related discussions):
\be
\mathcal{I}_6=\frac{N_c}{3!} \frac{1}{(2\pi)^3}\left(\text{tr} F_L^3- \text{tr} F_R^3\right)\;,
\ee
where $F_{L,R}$ denote the field strength two-forms of the background gauge fields coupled to $q_{L,R}$, respectively.
\item What remains unknown are the quantum numbers and the spectrum of putative hadrons in the IR --- the left-hand side of Eq.~\eqref{eq:AMC}. Without loss of generality, we need to sum over all the possible representations $r$ for the hadrons, and we assign for each representation $r$ with an index $\ell\!\left(r \right)$, which denotes how many times the representation $r$ appears in the spectrum with helicity $+1/2$ minus that of helicity $-1/2$.~\footnote{Notice that any parity-invariant spectrum of putative hadrons can satisfy the matching of $[U(1)_B]^3$ anomaly. Consider e.g. a hadron in a representation $r=(r_L, r_R, b)$ under $\mathcal{G}[N_f]$, the parity conjugate partner is in the representation $r_P=(r_R, r_L, b)$ and their helicities are opposite to each other, i.e., $\ell\!\left(r \right)=-\ell\!\left(r_P \right)$. Therefore, the contribution of the parity conjugate pair to the $[U(1)_B]^3$ anomaly is $\ell\!\left(r \right) b^3 + \ell\!\left(r_P \right) b^3=0$.}
\end{itemlist}

Here are some comments regarding the physical meaning of the indices $\ell\!\left(r \right)$'s:
\begin{itemlist}
 \item For a physical spectrum, all the indices must be integers, i.e., 
 \be
 \ell\!\left(r \right) \in \mathbb{Z} \quad \text{for\ any} \quad r.
 \ee
 Notice that negative integers are also allowed for $\ell\!\left(r \right)$. 
 \item Nontrivial indices, i.e., $|\ell\!\left(r \right)|>1$, imply \emph{emergent accidental symmetries} acting on the hadrons in the IR. This implies that, even though the symmetries in the UV and IR may be different, the 't Hooft anomalies of the UV symmetry must be matched in the IR. 
 \item The index $\ell\!\left(r \right)$ vanishes, i.e., $\ell\!\left(r \right)=0$, when the hadron in representation $r$ is vectorlike, namely the states of hadrons with opposite helicities are paired. It implies that, when all the indices vanish, one can build Dirac mass terms to fully gap the entire hadronic spectrum. A corollary is that 't Hooft AMC implies the theory cannot be trivially gapped when the anomalies in UV do not vanish. Indeed, when the right-hand side of Eq.~\eqref{eq:AMC} is not zero, not all the indices on the left-hand side of Eq.~\eqref{eq:AMC} can vanish, meaning that the unpaired chiral fermions remain massless.
\end{itemlist}
As a result, failure of matching 't Hooft anomalies with integral indices of color-singlet composite fermions suggests that unbroken chiral symmetry is incompatible with color screening condition in Eqs.~\eqref{color_screening} and~\eqref{color_screening_2}. In other words, $\chi$SB necessarily occurs if color screening is assumed; conversely, color charge cannot be fully screened in the unbroken phase of chiral symmetry. The challenge is then to prove that Eq.~\eqref{eq:AMC} does not have integral solutions of all $\ell\!\left(r \right)$ for \emph{any} putative spectrum of color-singlet hadrons in  \emph{any} QCD-like theories with generic $N_c$ and $N_f$. 

\begin{lemma}
When $N_c$ is even, all the color-singlet hadrons in the IR theory are bosons.
\end{lemma}
This is easily seen from the color screening condition in Eq.~\eqref{color_screening}: there are an even number of fermionic constituents in each hadron when $N_c$ is even, i.e.,
\be
n_q+n_{\bar{q}}=b N_c + 2 n_{\bar{q}} =0 \quad\text{mod}\quad 2\;.
\ee
Therefore, there is no putative spectrum of massless composite fermions to match the 't Hooft anomalies in theories with \emph{even} $N_c$, implying that chiral symmetry must be spontaneously broken once color screening is assumed.

As a result, to formulate a general proof, we only need to consider the theories with \emph{odd} $N_c$ and generic $N_f$.

\subsection{Algebraic Proof of Chiral Symmetry Breaking for Special $N_f$}
\label{sec:proof_1}

Among all the QCD-like theories with odd $N_c$, one can formulate an algebraic proof based on 't Hooft AMC showing that, for certain special values of $N_f$, spontaneous $\chi$SB necessarily occurs once color screening is assumed. Specifically, we prove the following result~\footnote{More precisely, what is proven in~[\citen{Ciambriello:2024xzd}] is that the $[SU(N_f)_{L,R}]^2 U(1)_B$ AMC equation admits no integral solutions when $N_f$ is proportional to a nontrivial prime factor of $N_c$, i.e. when $\text{gcd}(N_c, N_f)=p>1$. We see that this is slightly more general than the special $N_f$ considered as in Eq.~\eqref{eq:special_Nf}. See~[\citen{Preskill:1981sr, Weinberg:1996kr}] where a similar result for $N_c=3$ is stated by Preskill and Weinberg, but they did not provide a proof for their statement.}:
\begin{lemma}\label{thm:special_Nf}
Let $N^p_f$ be a nontrivial prime factor of $N_c$, i.e., 
\be
N_c =0 \quad \text{mod} \quad N^p_f\;, 
\label{eq:special_Nf}
\ee
the $[SU(N^p_f)_{L,R}]^2 U(1)_B$ AMC equation admits no integral solutions for the indices of any putative color-singlet hadrons.
\end{lemma}

To demonstrate the Lemma~\ref{thm:special_Nf}, let us consider a general representation 
\be
r=(r_L, r_R, b)
\ee
under the chiral symmetry group $\mathcal{G}[N^p_f]$, which characterizes a color-singlet massless hadron, whose contribution to the $[SU(N^p_f)_{L}]^2 U(1)_B$ anomaly coefficient is given by 
\be
\mathcal{A}(r)=T(r_L) d(r_R) b\;, 
\ee
where $T(r_L)$ denotes the Dynkin index of $r_L$ and $d(r_R)$ denotes the dimensionality of $r_R$. Here, our goal is to show that
\be
T(r_L) d(r_R) b =0 \quad\text{mod} \quad N^p_f \quad\quad \text{for\ \emph{any}}\quad r=(r_L, r_R, b)\;,
\label{eq:prime_factor}
\ee
such that only the \emph{fractional} indices $\ell\!\left(r \right)$ can solve the 't Hooft AMC equation $\sum_{r} \,\ell\!\left(r \right) \mathcal{A}\!\left(r\right)=N_c \times \frac{1}{N_c}=1$. 

Notice that the representation $r$ for hadrons has to satisfy the constraint from Eqs~\eqref{color_screening} and~\eqref{color_screening_2} due to color screening, which in turn implies that the $N_f$-alities of $r_L$ and $r_R$, denoted as $\mathcal{N}(r_L)$ and $\mathcal{N}(r_R)$~\footnote{The $N$-ality of a representation of the $SU(N)$ group counts the number of boxes of the corresponding Young tableau. For instance, $\mathcal{N}(\tiny\yng(1))=1$ and $\mathcal{N}(\bar{\tiny\yng(1)})=N-1$.}, need to satisfy
\be
\mathcal{N}(r_L)+\mathcal{N}(r_R)= n_q \cdot 1 + n_{\bar{q}} \cdot (N_f-1)= b N_c \quad \text{mod} \quad N_f\;.
\ee
As seen in Eq.~\eqref{eq:quark_charges}, each quark transforms as the fundamental representation under either $SU(N_f)_L$ or $SU(N_f)_R$, hence contributing $1$ to the $N_f$-alities; likewise, each antiquark transforms as the anti-fundamental representation under either $SU(N_f)_L$ or $SU(N_f)_R$, whose contribution to $N_f$-alities is $N_f-1$. For $N^p_f$, the above equation is simplified to
\be
\mathcal{N}(r_L)+\mathcal{N}(r_R) = 0  \quad \text{mod} \quad N^p_f\;. 
\label{eq:special_Nf_constraint}
\ee
It turns out that Eq.~\eqref{eq:prime_factor} can be proven by analyzing the individual values of $\mathcal{N}(r_L)$ and $\mathcal{N}(r_R)$, which are subject to the constraint in Eq.~\eqref{eq:special_Nf_constraint} that arises from the special choice of $N^p_f$ as a nontrivial prime factor of $N_c$.
More specifically, the proof of Lemma~\ref{thm:special_Nf} can be structured as follows:
 \begin{romanlist}[(ii)]
 \item \label{case1:prime_Nf} When $\mathcal{N}(r_L) =0$ mod $N^p_f$ (hence $\mathcal{N}(r_R) =0$ mod $N^p_f$), we show that 
 \be
 T(r_L)=0 \quad \text{mod} \quad N^p_f 
 \label{eq:case1:prime_Nf}
 \ee
 using the center generator of $SU(N^p_f)_L$, i.e., $T_{\text{center}}= \text{diag}\left(1,\cdots,1, -(p-1)\right)$. The main idea is to analyze the charge of $r_L$ under the $U(1)$ subgroup of $SU(N^p_f)_L$ generated by $T_{\text{center}}$.  
 \item \label{case2:prime_Nf} When $\mathcal{N}(r_L) \neq 0$ mod $N^p_f$ (hence $\mathcal{N}(r_R) \neq 0$ mod $N^p_f$), we show that 
 \be
 d(r_R)=0 \quad \text{mod} \quad N^p_f 
  \label{eq:case2:prime_Nf}
 \ee
 directly using the Weyl dimension formula (see e.g.~[\citen{Yamatsu:2015npn}]). 
\end{romanlist}
We refer the reader to~[\citen{Ciambriello:2024xzd}] for the detailed proof of~(\ref{case1:prime_Nf}) and~(\ref{case2:prime_Nf}), which we will not repeat here.

In the following, we illustrate the above two cases in Eqs.~\eqref{eq:case1:prime_Nf} and~\eqref{eq:case2:prime_Nf} with a concrete example, which we hope will help the reader build some intuition. (For more examples, see also~[\citen{Weinberg:1996kr, Preskill:1981sr}] and~[\citen{Ciambriello:2024msu}]~\footnote{Since parity is not spontaneously broken in QCD-like theories~[\citen{Vafa:1984xg}], all the examples considered in~[\citen{Ciambriello:2024msu}] consist in a parity-invariant spectrum, i.e. each color-singlet hadron has a partner under parity transformation, where the indices for a parity conjugate pair are equal and opposite. However, as seen in the example~\ref{exp:special_Nf}, the assumption of unbroken parity is not essential for the validity of Eq.~\eqref{eq:prime_factor}.}.)

\begin{example}\label{exp:special_Nf}
Let us consider the QCD-like theory with $N_c=3$ and, for instance, the following baryons in the putative spectrum:
\be
r_1 = \left(\tiny\yng(3)\ , s\ , 1\right)\;, \quad\quad\quad r_2=\left(\tiny\yng(2)\ , \tiny\yng(1)\ , 1\right)\; ,
\label{eq:spectrum_example1}
\ee
where each baryon contains three constituent quarks, hence the baryon number is $b=1$, and $s$ denotes the singlet representation. (Notice that $r_{1,2}$ cannot describe color-singlet hadrons when $N_c\neq 3$.)
The Dynkin indices and dimensions of the relevant representation of $SU(N_f)_{L,R}$ are~[\citen{Ciambriello:2024msu}] 
\be
T({\tiny\yng(3)}) = \frac{(N_f+3) (N_f+2)}{2}\;, \quad d(s)=1\;, \quad T({\tiny\yng(2)}) = N_f+2\;, \quad d({\tiny\yng(1)}) =N_f \;.
\label{eq:example_dynkin_dimension}
\ee
It is straightforward to observe that when $N^p_f=3$, 
\be
T({\tiny\yng(3)}) =0 \quad \text{mod} \quad 3\;, \quad\quad\quad d({\tiny\yng(1)}) =0 \quad \text{mod} \quad 3\;,
\ee
hence both $A(r_1)$ and $A(r_2)$ satisfy Eq.~\eqref{eq:prime_factor}, which match the two cases~(\ref{case1:prime_Nf}) and~(\ref{case2:prime_Nf}), respectively. 
\end{example}

We remark that, for any theory with specific $N_c$ and $N^p_f$ satisfying Eq.~\eqref{eq:special_Nf}, the result in Lemma~\ref{thm:special_Nf} may seem obvious and well-known for a specific putative spectrum, for instance as in Eq.~\eqref{eq:spectrum_example1}; however, the novelty of~[\citen{Ciambriello:2024xzd}] is that it provides a proof applicable to a general spectrum of \emph{any} putative massless composite fermions, as long as they are color-singlet hadrons. 
Since the dynamics of QCD-like theories cannot be resolved, it is necessary to consider all the possible color-singlet particles in order to have a rigorous proof.

\section{Persistent Mass Conditions}
\label{sec:pmc}
So far, we have shown that, for special values of $N_c$ and $N_f$ --- namely, when $N_c$ is even, or $N_c$ is odd and $N_f$ is a prime factor of $N_c$ --- AMC implies that unbroken chiral symmetry is incompatible with a fully color-screened, IR-free phase in a class of QCD-like theories. However, for other general values of $N_c$ and $N_f$, AMC alone is not sufficiently restrictive to determine the IR dynamics. In such cases, integral solutions to the AMC equations can be found~[\citen{Preskill:1981sr}], AMC alone is not conclusive. 

A natural question arises: can we identify additional constraints that, when combined with the AMC, allow for a definitive conclusion in QCD-like theories regarding the compatibility between unbroken chiral symmetry and a free hadronic phase with full color screening? 

The answer is provided by the so-called persistent mass conditions (PMC)~[\citen{tHooft:1979rat, Preskill:1981sr, Vafa:1983tf}]. 
Intuitively, the basic idea is to deform the massless theory by introducing finite quark masses while tracking the unbroken chiral symmetry along the deformation. This can be a powerful probe for the QCD-like theories, where gauge-invariant quark mass terms are allowed due to their vectorlike nature. PMC says that bound states with massive constituents must also be massive~\footnote{As we will see, such a statement needs a minor refinement: the bound states must also be charged under the $U(1)_{H_i}$ factors in the unbroken chiral symmetry; see e.g. after Eq~\eqref{eq:PMC_G1}.}, which was originally formulated by 't Hooft as the decoupling condition~[\citen{tHooft:1979rat}], later on reformulated by Preskill and Weinberg as PMC~[\citen{Preskill:1981sr}] such that potential phase transitions at large quark masses cannot invalidate these conditions, and finally proven by Vafa and Witten~[\citen{Vafa:1983tf}] with some mild assumptions (on continuity of parameters, as we will illustrate near the end of Section~\ref{sec:vw}). 

Some comments are as follows. 
\begin{itemlist}
\item PMC also implies that the vectorlike part of $\mathcal{G}[N_f]$ (see Eq.~\eqref{def:G}) cannot be spontaneously broken, a result commonly referred to as the Vafa-Witten theorem~[\citen{Vafa:1983tf}]. For instance, one can consider the limit that all the quark flavors have a common positive and real mass~\footnote{By definition, vectorlike symmetries are the ones preserved by quark mass terms. Hence, the maximal vectorlike symmetry is obtained in the limit where all the quark flavors are assigned a common mass.}, the vectorlike symmetry $SU(N_f)_V\times U(1)_B$ cannot be spontaneously broken; otherwise there are \emph{massless} Nambu-Goldstone bosons made from massive quarks, which contradicts the PMC.
\item Similar to Vafa-Witten's derivation on PMC, there are inequalities on hadron masses originally derived by Weingarten~[\citen{Weingarten:1983uj}]; see also~[\citen{Nussinov:1999sx}] for a comprehensive review. Those mass inequalities imply that, in QCD-like theories, the lightest hadrons in the spectrum are the mesons interpolated by the pseudo-density operator $\bar{q}\gamma^5 q$, i.e., 
\be
\langle 0 | \bar{q}\gamma^5 q | \pi \rangle \neq 0\;. 
\ee
If there are massless fermionic hadrons needed for AMC, there are also massless mesons, denoted as $\pi$, in the putative spectrum.

However, such an argument per se does not conclusively prove $\chi$SB, since it does not contradict the assumption of unbroken chiral symmetry. Indeed, when chiral symmetry is unbroken the axial current $J^\mu_A=\bar{q} \gamma^5\gamma^\mu T^a q$ and the pseudo-density operator $\bar{q}\gamma^5 q$ transform differently under $SU(N_f)_L\times SU(N_f)_R$~\footnote{More specifically, if the chiral symmetry group $SU(N_f)_L\times SU(N_f)_R\times U(1)_B$ is not spontaneously broken, $J^\mu_A=\bar{q} \gamma^5\gamma^\mu T^a q=\bar{q}_L \gamma^5\gamma^\mu T^a q_L+ \bar{q}_R \gamma^5\gamma^\mu T^a q_R$ transforms as the representation $(\text{Adj}, s, 0)+ (s, \text{Adj}, 0)$, while $\bar{q}\gamma^5 q=\bar{q}_L\gamma^5 q_R + \bar{q}_R\gamma^5 q_L$ transforms as $(\bar{\tiny\yng(1)}, \tiny\yng(1),0) + (\tiny\yng(1), \bar{\tiny\yng(1)}, 0)$. Here, $s$ stands for the singlet representation. Therefore, as we see explicitly, $\bar{q} \gamma^5\gamma^\mu T^a q$ transforms differently from $\bar{q}\gamma^5 q$ when chiral symmetry is unbroken.}, hence 
\be
\langle 0 | J^\mu_A | \pi \rangle = 0\;,
\ee
i.e., in this case the $\pi$ particles are \emph{not} Nambu-Goldstone bosons. (In contrast, if the $\pi$ particles were Nambu-Goldstone bosons in the broken phase of chiral symmetry, $\langle 0 | J^\mu_A | \pi \rangle \propto p^\mu$, where $p^\mu$ are the four-momentum of the $\pi$ particles.) Notice that, while QCD inequalities do not conclusively rule out the unbroken phase of chiral symmetry, they rule out the so-called Stern phase of the broken chiral symmetry~[\citen{Kogan:1998zc}].
\end{itemlist}

In the remainder of this section, we first briefly review Vafa and Witten's derivation of PMC, and then illustrate the algebraic structure (involving PMC equations with more than one massive flavor) that is crucial for our proof.
In particular, many results in Section~\ref{sec:algebra_PMC} were recently discovered in~[\citen{Ciambriello:2022wmh}]. 

\subsection{Derivation from Vafa and Witten}
\label{sec:vw}

There are two essential ingredients in the derivation of the PMC by Vafa and Witten in QCD-like theories~[\citen{Vafa:1983tf}]:
\begin{romanlist}
\item \emph{Positivity of the measure.} In vectorlike gauge theories with real and positive quark masses, the fermionic integration measure in the path integral is positive-definite. Consider the fermion partition function in a fixed background gauge field $A_\mu$:
\be
\int \mathcal{D}\bar q \ \mathcal{D}q \ e^{-\int d^4x \ \bar q (\slashed{D}+m) q } = \text{det}(\slashed{D}+m) \;.
\ee
In Euclidean space, the Dirac operator is anti-hermitian, i.e. it satisfies $\slashed{D}^\dagger=-\slashed{D}$, so its eigenvalues are purely imaginary:
\be
\slashed{D} \psi_\lambda = i\lambda \psi_\lambda \; .
\ee
Furthermore, since $\slashed{D}$ anti-commute with $\gamma_5$, i.e., $\{\slashed{D}, \gamma_5\}=0$, which implies that non-zero eigenvalues are paired: if $\lambda$ is an eigenvalue, $-\lambda$ is also, i.e., $\slashed{D} (\gamma_5\psi_\lambda) = -i\lambda (\gamma_5\psi_\lambda)$. The fermion determinant then becomes
\be
\text{det}(\slashed{D}+m)= \prod_{\lambda} (-i\lambda+m) = \prod_{\lambda>0} (\lambda^2+m^2)\cdot m^{n_0}>0 \;,
\ee
where $n_0$ is the number of the zero modes of the Dirac operator. Therefore, the fermionic measure is strictly positive for all $A_\mu$ provided $m>0$.

\item \emph{Exponentially-decaying upper bound on the quark propagator.} In a fixed gauge field background $A_\mu$, the quark propagator is given by 
\be
S^A_\Delta (x,y;m) \equiv (\slashed{D}+m)^{-1}_{x,y} \;.
\ee
Vafa and Witten showed that this propagator satisfies an exponential upper bound~\footnote{This is the Eq.~(25) of the original paper~[\citen{Vafa:1983tf}]. We refer the reader to the original paper for the derivation of this equation.}:
\be
| S^A_\Delta (x,y;m) |  \leq \alpha(\Delta, m) \ e^{-m |x-y|}\;,
\ee
where the subscript $\Delta$ denotes the smeared regions near the points $x$ and $y$, which can intuitively be thought of as a UV cutoff regularizing the short-distance behavior of the propagator. Here, what is crucial for us is that $\alpha(\Delta, m)$ is \emph{independent} of the background gauge field $A$, but only depends on the smearing scale $\Delta$ and the quark mass $m$. 
\end{romanlist}
Positivity of the fermionic measure ensures that the exponential bound survives after averaging over all gauge configurations. That is, the full gauge-averaged propagator in the theory with $N_f$ flavors
\be
\langle q(x) \bar q (y)\rangle=\frac{1}{Z} \int \mathcal{D} A_\mu \left[\text{det} (\slashed{D} +m)\right]^{N_f} e^{-S_{\text{YM}}[A]} S^A_\Delta (x,y;m)\;
\ee
satisfies the bound:
\be
|\langle q(x) \bar q (y)\rangle| \leq \alpha(\Delta, m) \ e^{-m |x-y|} \;.
\ee

More generally, one can derive a similar bound on the two-point function of the composite local operator $T(x)$, constructed from quark fields, that interpolates color-singlet hadrons. For instance, when all the flavors of quarks have an equal mass $m>0$, it follows that 
\be
|\langle T(x) T^\dagger (y)\rangle| \leq C(\Delta, m) \ e^{-m\cdot n |x-y|} \; ,
\label{eq:exp_decay}
\ee
where $n$ is the total number of quark and anti-quark propagators from the point $x$ to $y$. From another perspective, the composite operator $T(x)$ can interpolate a whole tower of color-singlet hadrons in the color-screening, IR-free phase with the same quantum number under the chiral symmetry. (Notice that so far we have not ruled out the option of unbroken chiral symmetry.) \emph{At large distances}, the two-function of $T(x)$ is dominated by the lightest hadron being exchanged in that channel, i.e.,
\be
\langle T(x) T^\dagger (y)\rangle \sim e^{-m_{\text{hadron}} |x-y|}\; ,
\label{eq:exp_decay_large_distance}
\ee
where all the other hadrons with larger masses contribute subleading terms that decay more rapidly and are exponentially suppressed.~\footnote{Here, all the hadrons interpolated by the composite local operator $T(x)$ are massive. If there were massless hadrons being excited in the same channel, the two-point function would decay with a \emph{power law} rather than an exponential law. If $\langle T(x) T^\dagger (y)\rangle$ obeys the power-law decay, it cannot satisfy the bound in Eq.~\eqref{eq:exp_decay} at large distances as $|x-y|\to \infty$.}
As a result, we obtain the following inequality by comparing Eqs.~\eqref{eq:exp_decay} and~\eqref{eq:exp_decay_large_distance}, which explicitly realizes the statement of PMC:
\be
m_{\text{hadron}} \geq m \cdot n >0 \;.
\ee

Following~[\citen{Vafa:1983tf}], one can also derive similar mass inequalities when assigning different masses to different quark flavors. For instance, let us consider the case where one quark flavor is assigned mass $m_1$, while the remaining flavors are assigned a common mass $m_2$. In this case, one can derive the following inequality for PMC
\be
m_{\text{hadron}} \geq m_1 \cdot n_1 + m_2 \cdot n_2 >0\;,
\ee 
where $n_{1,2}$ are the numbers of quark and anti-quark propagators from the point $x$ to $y$ associated with the quark flavors of the masses $m_{1,2}$, respectively. Notice that such a bound is valid for any finite UV cutoff associated with $\Delta$ and any positive quark masses $m_{1,2}$. (For instance, the UV cutoff can be extremely large as long as it remains finite, and $m_{1,2}$ can be infinitesimally small, such that the bound cannot be invalidated by any phase transitions at a small UV cutoff or large quark masses as compared to any dynamical scale.) 
\emph{Assuming} that the physical quantities should vary continuously with these parameters, the bound is expected to hold in the limits of $\Delta\to \infty$ (i.e., sending the UV cutoff to infinity), or $m_{1,2}\to 0$ (i.e., taking the massless quark limit).
Notice that, provided $n_{1,2}>0$, the lower bound for $m_{\text{hadron}}$ is strictly positive as long as $m_1$ and $m_2$ are not sent to zero simultaneously. For instance, assuming that $m_{\text{hadron}}$ is a continuous function of $m_2$, we have the PMC 
\be
m_{\text{hadron}} \geq m_1 \cdot n_1>0
\ee
in the limit $m_2\to 0$.

\subsection{Algebraic Structure of Persistent Mass Conditions}
\label{sec:algebra_PMC}
So far, we have established the inequalities for PMC following~[\citen{Vafa:1983tf}]. To apply these inequalities in our proof, we need to translate them into a system of equations involving the indices $\ell\!\left(r \right)$ associated wth the massless hadrons in the putative spectrum, which are required to match the 't Hooft anomalies. The basic idea is then to track the unbroken chiral symmetry within $\mathcal{G}[N_f]$ (see in Eq.~\eqref{def:G}) while turning on real and positive quark masses~[\citen{Ciambriello:2022wmh}]. (Throughout this section, all the quark masses are considered to be real and positive. In the following, we will not explicitly mention this point for brevity.)

For instance, we can turn on a mass $m_1$ for one quark flavor and a common mass $\epsilon$ for the other flavors. In the limit $\epsilon\to 0$, the chiral symmetry group $\mathcal{G}[N_f]$ is explicitly broken to  
\be
\mathcal{G}[N_f,1]=SU(N_f-1)_L \times SU(N_f-1)_R \times U(1)_B \times U(1)_{H_1}\, ,
\label{def:G_1}
\ee
where only the massive quark flavor with the mass $m_1$ is charged under $U(1)_{H_1}$. In the color-screened and IR-free phase, the composite operators $T(x)$ interpolate color-singlet hadrons which furnish various representations of $\mathcal{G}[N_f,1]$. Following the discussion in~\ref{sec:vw}, one can show a strictly positive lower bound on the masses of the hadrons with the $U(1)_{H_1}$ charge $H_1\neq 0$:
\be
m_{\text{hadrons}} \geq |H_1| m_1 >0\;.
\label{eq:PMC_G1}
\ee
Notice that it was crucial to assume $H_1\neq 0$ here, which ensures that the hadron contains (at least) $|H_1|$ quark (or anti-quark) propagators with the mass $m_1$ propagating from the point $x$ to $y$. In contrast, for the hadrons with $H_1=0$, no firm bound can be established on their masses in the limit $\epsilon\to 0$~[\citen{Ciambriello:2022wmh}]. Eq.~\eqref{eq:PMC_G1} implies that the index associated with each hadron transforming in a representation $r_1$ of $\mathcal{G}[N_f,1]$ with $H_1\neq 0$ must vanish, i.e., 
\be
0=\ell\!\left(r_1 \right) = \sum_{r} \ell\!\left(r \right) \kappa (r\to r_1)  \;.
\label{eq:PMC_G1_indices}
\ee
Here, the second equality follows from the decomposition of the representations: each representation $r_1$ of $\mathcal{G}[N_f,1]$ can appear in the decomposition of several representations $r$ of the full chiral symmetry group $\mathcal{G}[N_f]$, and accordingly, the indices $\ell\!\left(r_1 \right)$ can be computed from $\ell\!\left(r \right)$, where $\kappa (r\to r_1)$ counts how many times $r_1$ appears in the decomposition of $r$. We denote all the equations in form of Eq.~\eqref{eq:PMC_G1_indices} as PMC$[N_f,1]$,  where $N_f$ and $1$ denote the total number of flavors and that of the massive flavor with mass $m_1$, matching the notation of the unbroken chiral symmetry group as in Eq.~\eqref{def:G_1}.

Likewise, we can deform the massless theory by turning on quark masses $m_1$ and $m_2$ for the first two flavors and a common mass $\epsilon$ for the other flavors. When $m_1\neq m_2$ and taking the limit $\epsilon\to 0$, the unbroken chiral symmetry group is 
\be
\mathcal{G}[N_f,2]=SU(N_f-2)_L \times SU(N_f-2)_R \times U(1)_B \times U(1)_{H_1} \times U(1)_{H_2} \, ,
\label{def:G_2}
\ee
where the quark flavors with masses $m_{1,2}$ are charged under $U(1)_{H_1, H_2}$, respectively. 
For some reason that will become clear later (see Eq.~\eqref{eq:pmc_identification}), let us consider color-singlet hadrons with $H_1=0$ and $H_2\neq 0$ at the composite level, for which we establish the strictly positive lower bound on the masses in the limit $\epsilon\to 0$: 
\be
m_{\text{hadrons}} \geq |H_2| m_2 >0\;.
\label{eq:PMC_G2}
\ee
Again, no firm bound can be established in the limit $\epsilon\to 0$ for the hadrons with $H_1=H_2=0$. Eq.~\eqref{eq:PMC_G2} implies PMC$[N_f,2]$ for each hadron transforming in a representation $r_2$ of $\mathcal{G}[N_f,2]$ with $H_1=0$ and $H_2\neq 0$, i.e., 
\begin{align}
0=\ell\!\left(r_2 \right) &= \sum_{\hat{r}_1} \ell\!\left(\hat{r}_1 \right) \kappa (\hat{r}_1\to r_2) \nonumber\\
&= \sum_{\hat{r}_1} \left(\sum_{r} \ell\!\left(r \right) \kappa (r\to \hat{r}_1) \right) \kappa (\hat{r}_1\to r_2)  \;, 
\label{eq:PMC_G2_indices}
\end{align}
where $\hat{r}_1$ is a representation of $\mathcal{G}[N_f,1]$ with $H_1=0$, whose indices in turn can be computed from those of $r\in \mathcal{G}[N_f]$. Notice that $\ell\!\left(\hat{r}_1 \right)$ are \emph{not} sent to zero by PMC$[N_f,1]$, since $\hat{r}_1$ has the quantum number $H_1=0$.

\begin{figure}[t]
\centerline{\includegraphics[width=12cm]{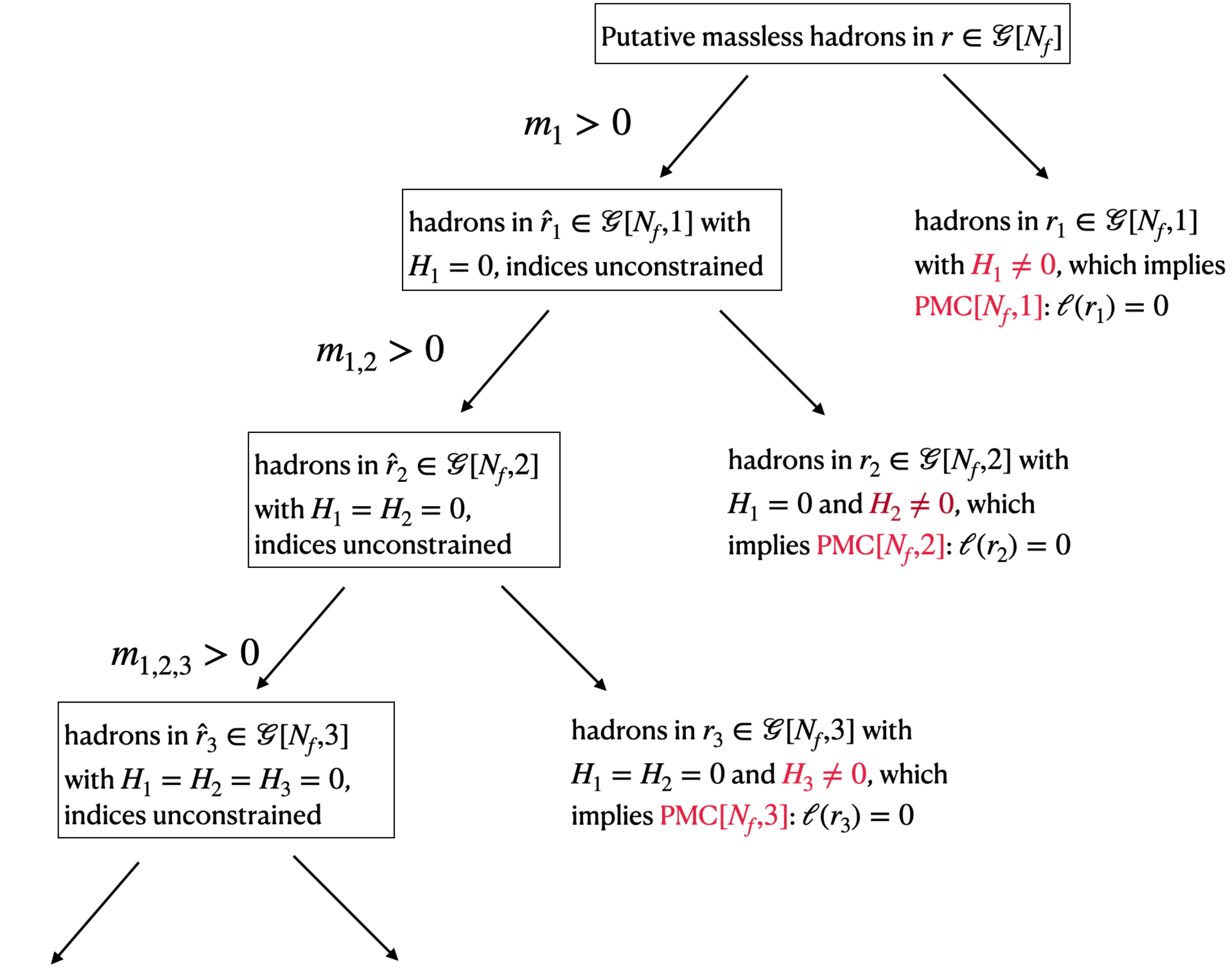}}
\caption{A schematic illustration of the derivation of various PMC$[N_f,h]$ equations from a chain of decomposition; see e.g. Eqs.~\eqref{eq:PMC_G1_indices},~\eqref{eq:PMC_G2_indices} and~\eqref{eq:PMC_Gh_indices}. The key idea is to deform the QCD-like theory of massless quarks by introducing real and positive quark masses, while carefully tracking the unbroken chiral symmetry group, in particular the associated nonzero $U(1)_{H_i}$ charges. All these mass deformations can be viewed as theoretical probes which, together with 't Hooft AMC, are useful for constraining the dynamics of the QCD-like theory.}
\label{PMC_decomposition}
\end{figure}

More generally, one can turn on different masses for the first $h$ quark flavors, $m_{1, 2, \cdots, h}$, and a common mass $\epsilon$ for the remaining flavors. In the limit $\epsilon\to 0$, the unbroken chiral symmetry group is
\be
\mathcal{G}[N_f,h]=SU(N_f-h)_L \times SU(N_f-h)_R \times U(1)_B \times U(1)_{H_1} \times \cdots \times U(1)_{H_h} \, ,
\label{def:G_h}
\ee
where the chiral part $SU(N_f-h)_L \times SU(N_f-h)_R$ remains nontrivial when $h\leq N_f-2$. Here, the quark flavor with mass $m_i$ is charged nontrivially under the $U(1)_{H_i}$ factor.
For the hadrons which are neutral under all the $U(1)_{H_i}$ factors with $i<h$ but charged nontrivially under $U(1)_{H_h}$, we find the strictly positive lower bound on the masses in the limit $\epsilon \to 0$:
\be
m_{\text{hadrons}} \geq |H_h| m_h >0\;,
\label{eq:PMC_Gh}
\ee
which implies PMC$[N_f, h]$ for each hadron transforming in a representation $r_h$ of $\mathcal{G}[N_f, h]$ with $H_i=0$ for all $i<h$ and $H_h\neq 0$. The indices $\ell\!\left(r_h \right)$ must vanish due to PMC$[N_f, h]$, and they can be computed from decomposing $r\in \mathcal{G}[N_f]$ in steps:
\begin{align}
0=\ell\!\left(r_h \right) &= \sum_{\hat{r}_{h-1}} \ell\!\left(\hat{r}_{h-1} \right) \kappa (\hat{r}_{h-1}\to r_h) \nonumber\\
&= \sum_{\hat{r}_{h-1}} \left(\sum_{\hat{r}_{h-2}} \ell\!\left(\hat{r}_{h-2} \right) \kappa (\hat{r}_{h-2}\to \hat{r}_{h-1}) \right) \kappa (\hat{r}_{h-1}\to r_h)  \nonumber\\
&=\cdots \;, 
\label{eq:PMC_Gh_indices}
\end{align}
where $\hat{r}_i$ is a representation of $\mathcal{G}[N_f,i]$ with vanishing $U(1)_{H_1} \times\cdots\times U(1)_{H_i}$ charges, whose indices are not sent to zero by the PMC in the previous steps of decompositions, but calculable from decomposing $r$. 

In Fig.~\ref{PMC_decomposition}, we summarize the whole process of decomposition that was used to derive all the PMC$[N_f,h]$ equations. 

\begin{figure}[t]
\centerline{\includegraphics[width=12cm]{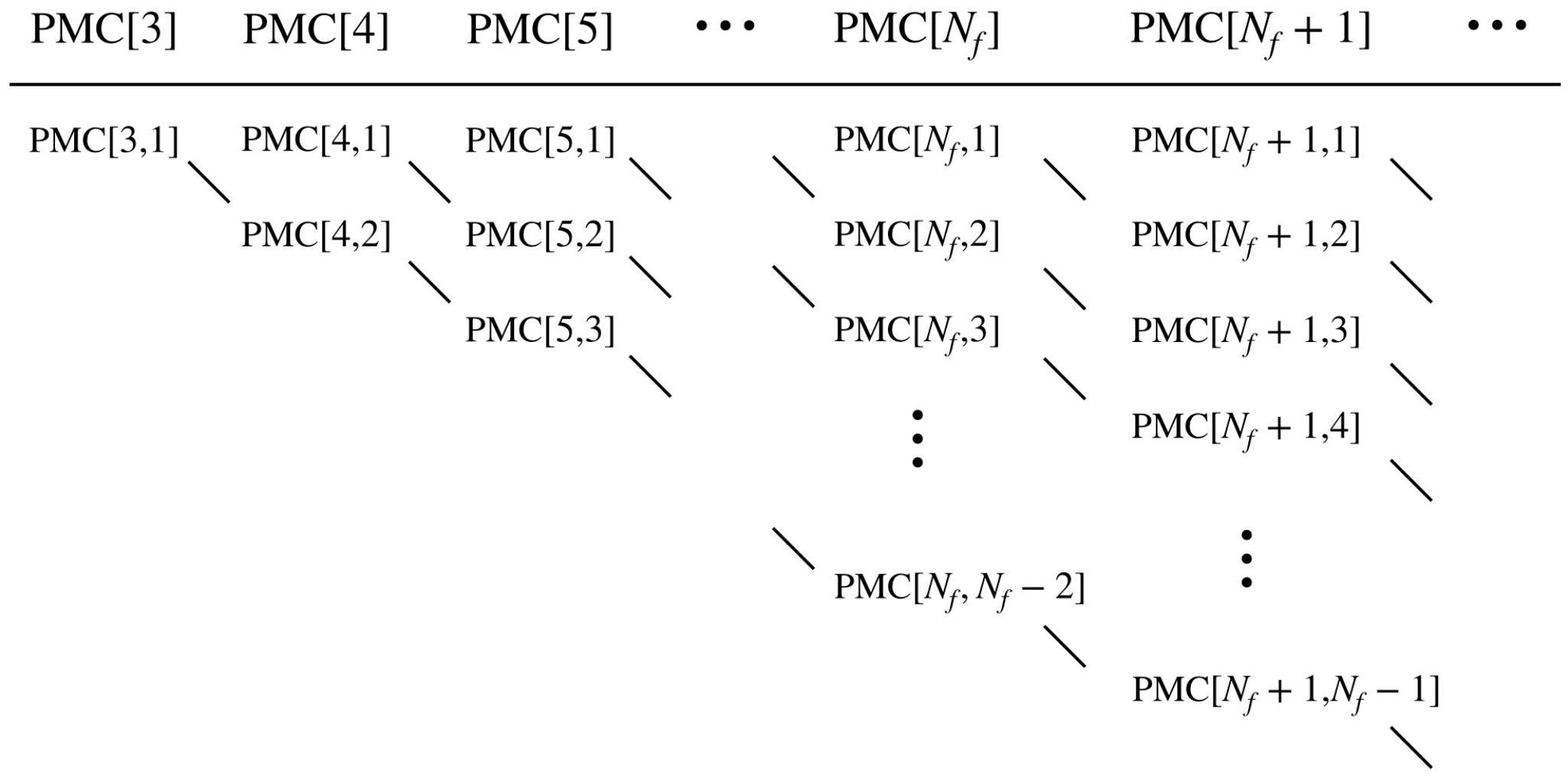}}
\caption{A coherent view of PMC equations for QCD-like theories with the same $N_c$ but different $N_f$. The PMC equations connected by a diagonal line have the same algebraic structure and can be identified. This arises from the fact that any representation $\hat{r}_1\in \mathcal{G}[N_f,1]$ with $H_1=0$ can effectively be viewed as a representation of $\mathcal{G}[N_f-1]$. For $2\leq h\leq N_f-2$, decomposing such a representation produces all PMC$[N_f, h]$ in the former theory of $N_f$ flavors and all PMC$[N_f-1, h-1]$ in the latter theory of $N_f-1$ flavors. Figure taken from~[\citen{Ciambriello:2022wmh}].}
\label{PMC_coherent}
\end{figure}

Having established all the PMC$[N_f, h]$ with $1\leq h\leq N_f-2$ for the QCD-like theory with $N_f$ flavors of massless quarks, it is useful to place them in a column as shown in Fig.~\eqref{PMC_coherent}. All the PMC equations in the same column are collectively denoted as PMC$[N_f]$.
One can repeat the same construction of PMC for the QCD-like theories with other values of $N_f$ (but with the same $N_c$). Here, the crucial observation is that any two sets of PMC equations connected by a diagonal line in Fig.~\eqref{PMC_coherent} have the same structure and thus can be identified~[\citen{Ciambriello:2022wmh}], i.e.,
\be
\text{PMC}[N_f, h] \sim \text{PMC}[N_f-1,h-1]\;, \quad \text{where} \quad 2\leq h\leq N_f-2\;.
\label{eq:pmc_identification}
\ee
This is because each $\hat{r}_1\in \mathcal{G}[N_f,1]$ with $H_1=0$ can effectively be viewed as a representation $r^\prime \in \mathcal{G}[N_f-1]$, i.e., 
\be
\hat{r}_1  \sim r^\prime \;,
\label{eq:pmc_rep_identification}
\ee
where the latter characterizes a putative color-singlet hadron in the QCD-like theory with the same $N_c$ but $N_f-1$ flavors of massless quarks~\footnote{Intuitively, one can compare $\mathcal{G}[N_f,1]=SU(N_f-1)_L \times SU(N_f-1)_R \times U(1)_B \times U(1)_{H_1}$ with $\mathcal{G}[N_f-1]=SU(N_f-1)_L \times SU(N_f-1)_R \times U(1)_B$, the difference is only the additional $U(1)_{H_1}$ factor for $\mathcal{G}[N_f,1]$. Since $\hat{r}_1$ is uncharged under this $U(1)_{H_1}$, it transforms in the same way under $\mathcal{G}[N_f,1]$ and $\mathcal{G}[N_f-1]$.}. Following the previous discussion in Fig.~\ref{PMC_decomposition}, decomposing the representation $\hat{r}_1$ produces all PMC$[N_f, h]$ with $h\geq 2$ in the theory with $N_f$ flavors, and decomposing $r^\prime$ produces all PMC$[N_f-1, h-1]$ in the theory with $N_f-1$ flavors, where $1\leq h-1\leq N_f-3$.

\section{Algebraic Proof of Chiral Symmetry Breaking for General $N_f$}
\label{sec:proof_2}

\subsection{Downlifting}
\label{sec:proof_2_downlifting}

In this section, we present a proof for Theorem~\ref{thm:main} by combining the arguments reviewed in Sections~\ref{sec:proof_1} and~\ref{sec:algebra_PMC}. The reader is encouraged to keep the results of those sections in mind as they proceed.

\begin{proof} 
The main strategy is referred to as ``\emph{Downlifting}''~[\citen{Ciambriello:2024xzd}], where the whole procedure is outlined as follows (see also summary in Fig.~\ref{downlifting}):
\begin{romanlist}[(ii)]
\item \label{downlift_step1} Suppose chiral symmetry is not spontaneously broken for the QCD-like theory with $N_c$ and $N_f$ while the free hadronic description with full color screening applies. This assumption implies that there must exist integral solutions of all AMC and PMC equations for the indices of putative color-singlet hadrons.
\item \label{downlift_step2} Given such an integral solution of AMC and PMC for $N_f$ flavors, one can construct an integral solution, denoted as the ``\emph{downlifted}'' solution, for the AMC and PMC equations for $N_f-1$ flavors (while keeping $N_c$ fixed). 

By iterating this procedure, we eventually obtain an integral solution of the AMC and PMC equations for $(N_f^p)_\text{min}$ flavors, where $(N_f^p)_\text{min}$ denotes the minimal nontrivial prime factor of $N_c$; see Eq.~\eqref{eq:special_Nf}.
\item \label{downlift_step3} However, by the Lemma~\ref{thm:special_Nf}, we know that no integral solution exists for the $[SU((N^p_f)_\text{min})_{L,R}]^2 U(1)_B$ AMC equation. This contradicts the initial assumption made in the step~(\ref{downlift_step1}). As a result, either the free hadronic description with full color-screening breaks down, or spontaneous $\chi$SB must occur in the QCD-like theory with $N_c$ colors and $N_f$ flavors.
\end{romanlist}
\end{proof}

\begin{figure}[t]
\centerline{\includegraphics[width=12cm]{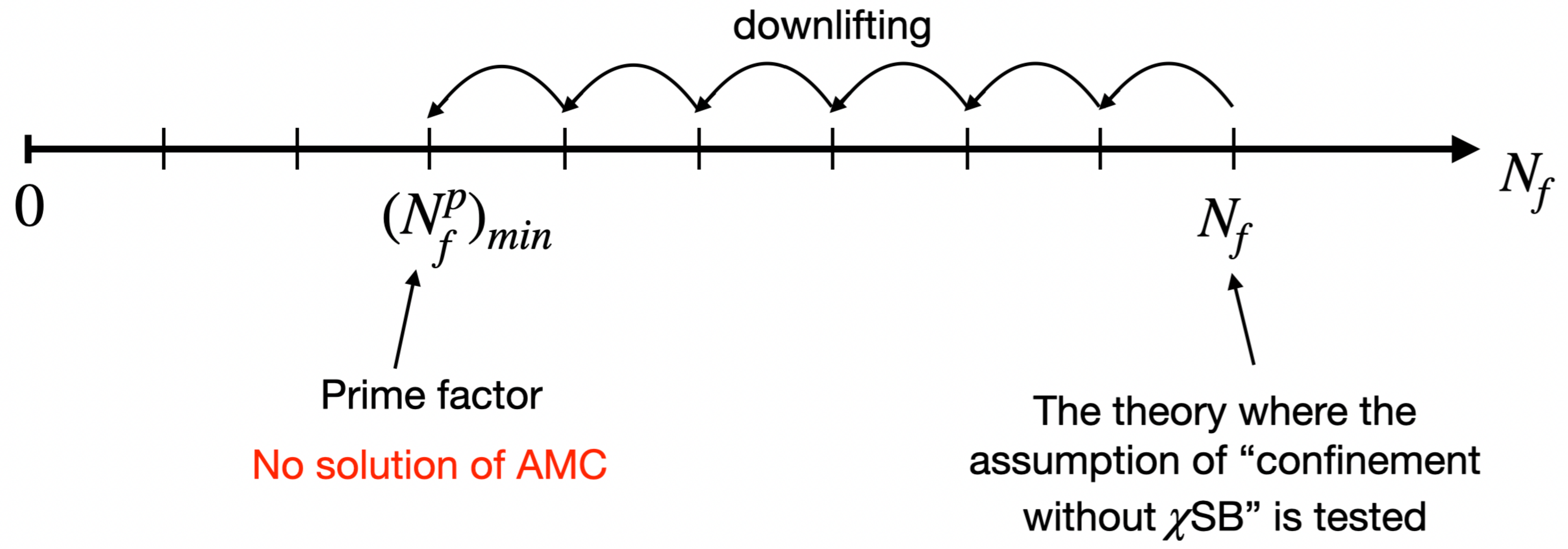}}
\caption{Illustration of the proof of Theorem~\ref{thm:main} based on the strategy of ``\emph{downlifting}''. The assumption of ``confinement without $\chi$SB'' in the QCD-like theory with $N_c$ colors and $N_f$ flavors implies that there exist integral solutions for AMC and PMC. Here, the term ``confinement'' refers to the screening confinement, as we explained in the Section~\ref{sec:intro}. Downlifting implies that there also exist integral solutions for AMC and PMC equations with $N_c$ colors and $N_f^\prime$ flavors, where $N_f^\prime<N_f$ which also includes $(N_f^p)_\text{min}$. However, there exists no integral solution for AMC with $N_c$ colors and $(N_f^p)_\text{min}$ flavors, falsifying the initial assumption. Figure adapted from~[\citen{Ciambriello:2024xzd}].}
\label{downlifting}
\end{figure}

As we have seen, the strategy of downlifting combines arguments by induction and contradiction. To finally complete the proof, it remains to rigorously justify the step~(\ref{downlift_step2}). For this purpose, we demonstrate the following:
\begin{lemma}\label{thm:downlifting}
Let the set of indices $\{\ell\!\left(r \right)\}$ be a solution of AMC and PMC for $N_f$ flavors, then a downlifted solution of AMC and PMC for $N_f-1$ flavors (while keeping $N_c$ fixed) can be constructed as:
\be
\ell\!\left(r^\prime \right) \equiv \sum_{r} \ell\!\left(r \right) \kappa (r\to r^\prime)\;,
\label{eq:thm:downlifting}
\ee
where $r$ and $r^\prime$ are the representations of $\mathcal{G}[N_f]$ and $\mathcal{G}[N_f-1]$, respectively, and $\kappa (r\to r^\prime)$ denotes the multiplicity of $r^\prime$ in the decomposition of $r$. 
\end{lemma}
 
\begin{proof}
The proof of the Lemma~\ref{thm:downlifting} mainly consists of using the coherent structure of PMC that was discussed in Section~\ref{sec:algebra_PMC}.

In the QCD-like theory with $N_c$ colors and $N_f$ flavors, the putative color-singlet hadrons are characterized by representations $r\in \mathcal{G}[N_f]$, whose indices $\ell\!\left(r \right)$ are the solutions of AMC and PMC equations for $N_f$ flavors when chiral symmetry is assumed unbroken. 

When giving mass to one flavor, each $r\in \mathcal{G}[N_f]$ is decomposed into $r_1$ and $\hat{r}_1\in \mathcal{G}[N_f,1]$, where $r_1$ and $\hat{r}_1$ carry $H_1\neq 0$ and $H_1=0$ charges under the $U(1)_{H_1}$ subgroup, respectively. The indices of $r_1$ and $\hat{r}_1$ are calculable from those of $r$. Specifically, the indices of $r_1$ vanish as in PMC$[N_f,1]$ (see Eq.~\eqref{eq:PMC_G1_indices}), while the indices of $\hat{r}_1$ are given by
\be
\ell\!\left(\hat{r}_1 \right) = \sum_{r} \ell\!\left(r \right) \kappa (r\to \hat{r}_1)\;.
\label{eq:thm:downlifting_2}
\ee
The indices $\ell\!\left(\hat{r}_1 \right)$ do not vanish, but they solve PMC$[N_f, h]$ equations with $2\leq h\leq N_f-2$, which are obtained by further decomposing $\hat{r}_1$ with more than one massive flavor; see Fig.~\ref{PMC_decomposition}. 

It implies that, due to the identification between the PMC equations as in Eq.~\eqref{eq:pmc_identification}, and accordingly, between each $\hat{r}_1\in \mathcal{G}[N_f,1]$ and $r^\prime\in \mathcal{G}[N_f-1]$ as in Eq.~\eqref{eq:pmc_rep_identification}, the indices constructed as in Eq.~\eqref{eq:thm:downlifting} solve PMC$[N_f-1, h-1]$ equations with $2\leq h\leq N_f-2$, which are all the PMC equations for the theory with $N_c$ colors and $N_f-1$ flavors. 

Next, we show that the indices constructed as in Eq.~\eqref{eq:thm:downlifting} also solve AMC for $N_c$ colors and $N_f-1$ flavors. 

Let us start by evaluating the anomaly coefficients of the $[SU(N_f)_{L,R}]^3$ or $[SU(N_f)_{L,R}]^2 U(1)_B$ triangles using the generators of the $SU(N_f-1)_{L,R}$ Lie subalgebras, where each irreducible representation of $SU(N_f)_{L,R}$ is now reducible under the $SU(N_f-1)_{L,R}$ subalgebra, hence it can be decomposed into a direct sum of a series of irreducible representations of $SU(N_f-1)_{L,R}$. Along such a decomposition, we rewrite the anomaly coefficients $\mathcal{A}(r)$ in Eq.~\eqref{eq:AMC} as follows,  
\be
\mathcal{A}(r)=\sum_{\hat{r}_1} \kappa(r\to \hat{r}_1) \mathcal{A}(\hat{r}_1) +\sum_{r_1} \kappa(r\to r_1) \mathcal{A}(r_1) \ ,
\label{eq:example_concerned}
\ee
where on the right-hand side the sum runs over all the representations in the decomposition. The relevant representations coincide with those in $\hat{r}_1, r_1 \in \mathcal{G}[N_f,1]$, and accordingly, $\mathcal{A}(\hat{r}_1)$ and $\mathcal{A}(r_1)$ stand for their anomaly coefficients of either the $[SU(N_f-1)_{L,R}]^3$ or $[SU(N_f-1)_{L,R}]^2 U(1)_B$ triangles. Plugging this equation back to Eq.~\eqref{eq:AMC} (i.e., the AMC equation for $N_c$ colors and $N_f$ flavors) and switching the order of the sums, we obtain
\begin{align}
\mathcal{A}(q) &= \sum_r \ell\!\left( r \right) \left(\sum_{\hat{r}_1} \kappa(r\to \hat{r}_1) \mathcal{A}(\hat{r}_1) +\sum_{r_1} \kappa(r\to r_1) \mathcal{A}(r_1) \right) \nonumber\\
&= \sum_{\hat{r}_1} \left(\sum_r  \ell\!\left( r \right) \kappa(r\to \hat{r}_1) \right) \mathcal{A}(\hat{r}_1) + \sum_{r_1} \left(\sum_r  \ell\!\left( r \right) \kappa(r\to r_1) \right) \mathcal{A}(r_1)   \nonumber \\
&= \sum_{\hat{r}_1} \ell\!\left(\hat{r}_1 \right) \mathcal{A}(\hat{r}_1)\;,
\label{eq:AMC_downlifting}
\end{align}
where in the last step we have used Eq.~\eqref{eq:thm:downlifting_2} for $\hat{r}_1$, and crucially the PMC$[N_f,1]$, namely Eq.~\eqref{eq:PMC_G1_indices}, which sent the indices of all $r_1$ to zero. 

Again, due to the identification between representations $\hat{r}_1\in \mathcal{G}[N_f,1]$ and $r^\prime\in \mathcal{G}[N_f-1]$ as in Eq.~\eqref{eq:pmc_rep_identification}, Eq.~\eqref{eq:AMC_downlifting} is nothing but the AMC equation for $N_c$ colors and $N_f-1$ flavors, i.e.,
\be
\mathcal{A}(q) = \sum_{r^\prime} \ell\!\left(r^\prime \right) \mathcal{A}(r^\prime)
\ee
whose solution is given by the indices constructed as in Eq.~\eqref{eq:thm:downlifting}.
\end{proof}

We emphasize that the proof crucially relies on the coherent structure shown in Fig.~\ref{PMC_coherent}, which was only discovered recently~[\citen{Ciambriello:2022wmh}]. For numerous examples illustrating the strategy of ``downlifting'', see~[\citen{Ciambriello:2024msu}]. We comment that once a specific putative spectrum is assumed, it may be more straightforward to solve the AMC and PMC directly and verify explicitly that no integral solutions exist for the indices of the color-singlet hadrons. In such cases, downlifting may seem unnecessary. However, since we cannot determine whether a particular hadron dynamically forms or not, it is essential that we make no assumptions on the spectrum in order to rigorously prove Theorem~\ref{thm:main}. Indeed, the novelty of downlifting~[\citen{Ciambriello:2024xzd}] is to provide a general proof applicable to a general spectrum of \emph{any} putative massless composite fermions, as long as they are color-singlet hadrons.

\begin{example}
In the proof of Lemma~\ref{thm:downlifting}, Eq.~\eqref{eq:example_concerned} may look a bit abstract, and it may need further clarification. In general, it means that the anomaly coefficients match before and after decomposition, whose validity is guaranteed at the group-theoretical level. 

Here, we illustrate this equation with an example. Let us consider the putative baryon in the representation 
\be
r_2=\left(\tiny\yng(2)\ , \tiny\yng(1)\ , 1\right) \quad\in\quad  \mathcal{G}[N_f]
\ee
which is valid for $N_c=3$ and any $N_f\geq 3$ (see also Eq.~\eqref{eq:spectrum_example1}). One can decompose the $SU(N_f)_{L,R}$ representations into the direct sum of a series of $SU(N_f-1)_{L,R}$ representations. The relevant decompositions for us are
\begin{align}
\tiny{\yng(2)} &\to \tiny{\yng(2)} + \tiny{\yng(1)} + s\;, \\
\tiny{\yng(1)} &\to \tiny{\yng(1)} + s \;,
\end{align}
where $s$ stands for the singlet representation of either $SU(N_f)$ or $SU(N_f-1)$.
Let us check Eq.~\eqref{eq:example_concerned} by evaluating the $[SU(N_f)_L]^2 U(1)_B$ anomaly coefficients for $r_2$ first using the Dynkin index and dimensionality under $SU(N_f)_{L,R}$: 
\be
\text{LHS}=T({\tiny\yng(2)}) d({\tiny\yng(1)}) = (N_f+2) N_f, 
\ee
and then using those of the $SU(N_f-1)_{L,R}$ subalgebra:  
\begin{align}
\text{RHS}&=T(\tiny{\yng(2)}) d(\tiny{\yng(1)}) + T(\tiny{\yng(2)}) d(s) +  T(\tiny{\yng(1)}) d(\tiny{\yng(1)}) + T(\tiny{\yng(1)}) d(s) + T(s) d(\tiny{\yng(1)}) + T(s) d(s) \nonumber \\
&= (N_f+1) (N_f-1) + (N_f+1) + (N_f-1) + 1 +0 + 0\nonumber \\
&= (N_f+2) N_f\;. 
\end{align}
As expected, the two results match. The same applies to other representations.
\end{example}

\section{Concluding Remarks}
\label{sec:discussion}

Tackling emergent phenomena in strongly coupled systems remains one of the most challenging tasks in both high-energy physics and condensed matter physics. The prominent example related to particle physics is QCD, where the short-distance theory is described by a gauge theory of quarks and gluons, while the emergent long-distance theory is described by color-singlet hadrons, which are composite particles formed due to color screening. 

In this review, we study such an emergent phenomenon of color screening and its relation with $\chi$SB in a class of QCD-like theories with generic $N_c$ and $N_f$, where the main result was already summarized in Theorem~\ref{thm:main}. 
\begin{itemlist}
\item~\label{arg1} On the one hand, it states that the chiral symmetry group $SU(N_f)_L\times SU(N_f)_R \times U(1)_B$ must be spontaneously broken once the theory flows to an IR-free fixed point with full color screening. 
\item~\label{arg2} On the other hand, the original Vafa-Witten theorem~[\citen{Vafa:1983tf}] states that the maximal vectorlike subgroup $SU(N_f)_V\times U(1)_B$ is not spontaneously broken. 
\end{itemlist}
Consequently, combining these two arguments, we see that the low-energy description of a QCD-like theory in the color-screened, IR-free phase is given by the non-linear sigma model (NLSM) with the symmetry breaking pattern
\be
SU(N_f)_L\times SU(N_f)_R \to SU(N_f)_V. 
\ee
Accordingly, there are massless particles known as the Nambu-Goldstone bosons (i.e., pions) that correspond to the broken generators. In the NLSM, the crucial ingredient for matching the 't Hooft anomaly of broken generators is the WZW term~[\citen{Wess:1971yu, Witten:1983tw}]; see also~[\citen{Lee:2020ojw, Yonekura:2020upo, Yonekura:2019vyz}]. 

While the final result has been established long ago as the conventional wisdom, recent works~[\citen{Ciambriello:2022wmh, Ciambriello:2024xzd, Ciambriello:2024msu}] provide a proof, thereby elevating it to the status of a rigorous (physical) theorem. It is worth emphasizing that Theorem~\ref{thm:main} concerns only the assumption that the theory in the IR is described by a free fixed-point with color-singlet hadrons. This assumption is itself unambiguous and does not rely on potentially subtle definitions of confinement in the presence of matter in the fundamental representation.~\footnote{That is, whether or not the color-singlet spectrum of hadrons with integer baryon numbers constituents a valid definition of confinement is a separate issue, the conventional wisdom of Higgs-confinement continuity says that the answer is no. See however a perspective from strip algebra~[\citen{Gagliano:2025gwr}].} 

Nonetheless, major challenges remain. Still, we are not able to coherently understand the IR phases of QCD-like theories for specific $N_c$ and $N_f$. More technically, a generalization of Theorem~\ref{thm:main} is still needed for $N_f$ smaller than the smallest prime factor of $N_c$. For example, in theories resembling the real-world QCD --- that is, with $N_c=3$ --- Theorem~\ref{thm:main} applies to all $N_f\geq 3$, but not to the case $N_f=2$. Furthermore, the proof crucially relies on the use of PMC, which is only valid in vectorlike theories. As a result, statements analogous to Theorem~\ref{thm:main} can be generalized to other vectorlike theories with different gauge groups and matter fields, but not to the theories where PMC cannot be applied, e.g., the chiral gauge theories.

Concerning the chiral gauge theories --- such as the SM of particle physics --- our understanding remains incomplete, particularly in the strongly-coupled regime. The primary challenge lies in the lack of reliable computational tools to study these theories. While lattice formulations of QCD and in general vectorlike theories~[\citen{Wilson:1974sk, Kogut:1974ag, Susskind:1976jm}] are very successful, regularizing chiral fermions on the lattice remains a subtle issue due to the Nielsen-Ninomiya theorem~[\citen{Nielsen:1980rz, Nielsen:1981xu, Nielsen:1981hk}]. For pedagogical introductions and reviews on chiral symmetries and anomalies on the lattice, see e.g.~[\citen{Chandrasekharan:2004dph, Kaplan:2009yg, Creutz:2011ati}], and for a lattice formulation of the SM, see~[\citen{Wang:2018jkc}]. An alternative computational tool may be the functional renormalization group approach~[\citen{Li:2025tvu}]. Another promising direction involves the use of 't Hooft anomalies of generalized global symmetries; see e.g.~[\citen{Bolognesi:2021jzs, Bolognesi:2023sxe, Konishi:2024rjz}] for reviews on the recent developments.

\section*{Acknowledgments}
I am deeply indebted to Luca Ciambriello, Andrea Luzio, and Marcello Romano, and in particular Roberto Contino, for their many insights and for the fruitful collaboration that led to~[\citen{Ciambriello:2022wmh, Ciambriello:2024xzd, Ciambriello:2024msu}]. 
I thank Mohamed M. Anber, Soo-Jong Rey, Kazuki Sakurai, and Ariel R. Zhitnitsky for useful comments and questions during a workshop held at Mainz Institute for Theoretical Physics (MITP) of the Cluster of Excellence PRISMA+ (Project ID 390831469).
I thank Aleksey Cherman, Theo Jacobson, and Yuya Tanizaki for their useful comments on defining confinement in the presence of fundamental matter during a workshop ``Bridging analytical and numerical methods for quantum field theory'' held at ECT$^*$ Trento.
I also thank the audiences at ITP-CAS Beijing, MITP, UC Louvain, Scuola Normale Superiore in Pisa, for feedbacks when some of the results were presented, and referees of~[\citen{Ciambriello:2022wmh, Ciambriello:2024xzd, Ciambriello:2024msu}] for useful comments; these inputs were helpful in shaping the structure of the present review. 
Last but not least, I am honored to thank Prof. K K Phua for the invitation to write this review.
This work is partially supported by ``Exotic High Energy Phenomenology" (X-HEP), a project funded by the European Union - Grant Agreement n.101039756. Funded by the European Union, views and opinions expressed are however those of the author(s) only and do not necessarily reflect those of the European Union or the ERC Executive Agency (ERCEA). Neither the European Union nor the granting authority can be held responsible for them.
This research was supported by the Munich Institute for Astro-, Particle and BioPhysics (MIAPbP) which is funded by the Deutsche Forschungsgemeinschaft (DFG, German Research Foundation) under Germany's Excellence Strategy – EXC-2094 – 390783311.

\section*{ORCID}
\noindent Ling-Xiao Xu - \url{https://orcid.org/0000-0002-4970-2404}
\appendix

\section{Brief Comparison with Other Approaches}
\label{app1:other_approaches}
In this appendix, we make a few brief remarks by comparing the results reviewed in Section~\ref{sec:proof_2_downlifting} with some other existing findings in the literature. 

Different from the downlifting approach, early works~[\citen{Frishman:1980dq, Farrar:1980sn, Schwimmer:1981yy, Cohen:1981iz, Kaul:1981fd, Takeshita:1981sx}] have attempted to establish an one-to-one correspondence in Eq.~\eqref{eq:thm:downlifting} between $r\in \mathcal{G}[N_f]$ and $r^\prime \in \mathcal{G}[N_f-1]$ (accordingly, each representation in their decompositions as well), a line of reasoning referred to as ``\emph{$N_f$ independence}''. Here are some comments~[\citen{Ciambriello:2022wmh}].
\begin{itemlist}
\item That is, when ``$N_f$ independence'' holds, there exists only one $r$ that can survive in the sum on the right-hand side of Eq.~\eqref{eq:thm:downlifting} for each $r^\prime$. \item Another subtlety concerns equivalent tensors: hadrons in the same representation of the chiral symmetry group can be interpolated by different composite operators with the same quantum number. However, such an equivalence does not necessarily persist as $N_f$ varies while $N_c$ is fixed, leading to another obstruction to ``$N_f$ independence''. 
\end{itemlist}
As shown in the proof~[\citen{Ciambriello:2022wmh}] and numerous examples in~[\citen{Ciambriello:2022wmh, Ciambriello:2024msu}], this line of reasoning is valid only under some specific assumption on the putative spectrum relative to the values of $N_c$ and $N_f$:
\be
(n_q+n_{\bar{q}})_{\text{max}}= (b N_c +2 n_{\bar{q}})_{\text{max}} <N_f \;,
\ee
where $(n_q)_{\text{max}}$ and $(n_{\bar{q}})_{\text{max}}$ are the maximal numbers of quark and anti-quark constituents in any hadron in the putative spectrum, respectively. We refer the reader to~[\citen{Ciambriello:2022wmh}] for the derivation of this equation.
Clearly, since we do not have dynamical control over the formation of putative hadrons, ``$N_f$ independence'' fails as a general proof of Theorem~\ref{thm:main}.

Another interesting line of reasoning relies on the $SU(N_f|N_f)$ \emph{superalgebra}, where the Lie algebra of the chiral symmetry group in Eq.~\eqref{def:G} is a sub-algebra generated by the even generators of $SU(N_f|N_f)$~[\citen{Schwimmer:1981yy}]. (See also~[\citen{BahaBalantekin:1981fmm, Bars:1982ps, Bars:1984rb}] for reviews.) Here are some remarks.
\begin{itemlist}
\item The key insight is that the grading of the superalgebra is defined so that the odd generators mix together opposite fermion helicities. As a result, when one flavor is assigned a mass, representations of unbroken $SU(N_f-1|N_f-1)$ superalgebra containing the massive flavor come in pairs with opposite grading; accordingly, the representations of $\mathcal{G}[N_f,1]$ (see Eq.~\eqref{def:G_1}) come in pairs with opposite helicities, hence PMC$[N_f,1]$ are trivially satisfied by \emph{any} single representation of $SU(N_f|N_f)$~[\citen{Schwimmer:1981yy}]. 
\item Such an argument is elegant, but is still subject to the complication from equivalent tensors of $\mathcal{G}[N_f,1]$, which cannot be captured by $SU(N_f-1|N_f-1)$. This is because the fully-antisymmetric tensor $\epsilon_{i_1 i_2\cdots i_{N_f-1}}$ is only an invariant tensor of $SU(N_f-1)$ Lie algebra, but not for the $SU(N_f-1|N_f-1)$ superalgebra~[\citen{BahaBalantekin:1981fmm, Bars:1982ps, Bars:1984rb}]. As a result, it appears that superalgebra representations provide solutions to PMC, but evidence is still lacking whether \emph{all} PMC solutions can be captured by the superalgebra.
\end{itemlist}
To better illustrate the second point, let us consider the following example.

\begin{example}
Consider the theory with $N_c=5$ and $N_f=3$, with the interpolating composite operators characterized by the following tensors
\bea
T_1=(\tiny\yng(3,1), s, 1, 1)\ ,\quad\quad\quad T_2=(\tiny\yng(2),\tiny\yng(1,1), 1 , 1)\;,
\eea
which are equivalent under $\mathcal{G}[3,1]=SU(2)_L\times SU(2)_R\times U(1)_{H_1} \times U(1)_B$. Here, $s$ stands for the singlet representation of $SU(2)_{L,R}$. These tensors characterize baryons with the baryon number $b=1$ when one quark flavor is assigned a real and positive mass. 
\begin{itemlist}
\item Since they also carry non-vanishing charges under $U(1)_B$, their indices are subject to the constraint of PMC$[3,1]$, which requires $\ell(T_1)+\ell(T_2)=0$.
It is possible to have $\ell(T_1)=-\ell(T_2)\neq 0$ and build a Dirac mass term by pairing $T_1$ and $T_2$ since they are equivalent. Notice that in this example, equivalent tensors are closely related to the two-index fully antisymmetric tensor $\epsilon_{i_1 i_2}$, which is an invariant tensor of $SU(2)$, and correspondingly, a column of two boxes in the Young diagram is a singlet. 

\item From the point of view of the superalgebra, $T_1$ and $T_2$ come from the decomposition of two independent irreducible representations of $SU(2|2)$, e.g., $\tiny\young(\diagup\diagup\diagup,\diagup)$ and $\tiny\young(\diagup\diagup\diagup\diagup)$, where each of them is in a pair with opposite grading. (Here we follow the same notation of Young diagrams as in~[\citen{BahaBalantekin:1981fmm, Bars:1982ps, Bars:1984rb}].) Accordingly, $T_1$ and $T_2$ come in separate pairs of opposite helicities. 
Hence, the superalgebra solution requires $\ell(T_1)=0$ and $\ell(T_2)=0$.
\end{itemlist}
In this example, we see that when equivalent tensors appear, individual superalgebra equations are more restrictive than the corresponding PMC equations. It would be quite surprising if eventually superalgebra yields exactly the same solutions as PMC. 
\end{example}

Compared to either $N_f$ independence or superalgebra, perhaps the main advantage of downlifting is that it does not rely on additional dynamical assumptions on the putative spectrum of massless hadrons. As a result, to the best of our knowledge, only downlifting can furnish a general proof of Theorem~\ref{thm:main}.


\bibliographystyle{ws-ijmpa}
\bibliography{refs}

\begin{thebibliography}{100}
\expandafter\ifx\csname urlstyle\endcsname\relax
  \providecommand{\doi}[1]{doi:\discretionary{}{}{}#1}\else
  \providecommand{\doi}{doi:\discretionary{}{}{}\begingroup
  \urlstyle{rm}\Url}\fi

\bibitem{ORaifeartaigh:1986agb}
L.~O'Raifeartaigh, {\em {GROUP STRUCTURE OF GAUGE THEORIES}}Cambridge
  Monographs on Mathematical Physics, Cambridge Monographs on Mathematical
  Physics (Cambridge University Press, 5 1988).

\bibitem{Hucks:1990nw}
J.~Hucks, {\em Phys. Rev. D} {\bf 43}, 2709  (1991),
  \doi{10.1103/PhysRevD.43.2709}.

\bibitem{Tong:2017oea}
D.~Tong, {\em JHEP} {\bf 07},   104  (2017),
  \href{http://arxiv.org/abs/1705.01853}{{\ttfamily arXiv:1705.01853
  [hep-th]}}, \doi{10.1007/JHEP07(2017)104}.

\bibitem{Davighi:2019rcd}
J.~Davighi, B.~Gripaios and N.~Lohitsiri, {\em JHEP} {\bf 07},   232  (2020),
  \href{http://arxiv.org/abs/1910.11277}{{\ttfamily arXiv:1910.11277
  [hep-th]}}, \doi{10.1007/JHEP07(2020)232}.

\bibitem{Wan:2019gqr}
Z.~Wan and J.~Wang, {\em JHEP} {\bf 07},   062  (2020),
  \href{http://arxiv.org/abs/1910.14668}{{\ttfamily arXiv:1910.14668
  [hep-th]}}, \doi{10.1007/JHEP07(2020)062}.

\bibitem{Wang:2021ayd}
J.~Wang, Z.~Wan and Y.-Z. You, {\em Phys. Rev. D} {\bf 106},   L041701  (2022),
  \href{http://arxiv.org/abs/2112.14765}{{\ttfamily arXiv:2112.14765
  [hep-th]}}, \doi{10.1103/PhysRevD.106.L041701}.

\bibitem{Wang:2020mra}
J.~Wang, {\em Phys. Rev. D} {\bf 103},   105024  (2021),
  \href{http://arxiv.org/abs/2012.15860}{{\ttfamily arXiv:2012.15860
  [hep-th]}}, \doi{10.1103/PhysRevD.103.105024}.

\bibitem{Alonso:2024pmq}
R.~Alonso, D.~Dimakou and M.~West, {\em Phys. Lett. B} {\bf 863},   139354
  (2025), \href{http://arxiv.org/abs/2404.03438}{{\ttfamily arXiv:2404.03438
  [hep-ph]}}, \doi{10.1016/j.physletb.2025.139354}.

\bibitem{Li:2024nuo}
H.-L. Li and L.-X. Xu, {\em JHEP} {\bf 07},   199  (2024),
  \href{http://arxiv.org/abs/2404.04229}{{\ttfamily arXiv:2404.04229
  [hep-ph]}}, \doi{10.1007/JHEP07(2024)199}.

\bibitem{Koren:2024xof}
S.~Koren and A.~Martin, {\em SciPost Phys.} {\bf 18},   004  (2025),
  \href{http://arxiv.org/abs/2406.17850}{{\ttfamily arXiv:2406.17850
  [hep-ph]}}, \doi{10.21468/SciPostPhys.18.1.004}.

\bibitem{Frohlich:1980gj}
J.~Frohlich, G.~Morchio and F.~Strocchi, {\em Phys. Lett. B} {\bf 97}, 249
  (1980), \doi{10.1016/0370-2693(80)90594-8}.

\bibitem{Frohlich:1981yi}
J.~Frohlich, G.~Morchio and F.~Strocchi, {\em Nucl. Phys. B} {\bf 190}, 553
  (1981), \doi{10.1016/0550-3213(81)90448-X}.

\bibitem{tHooft:1979yoe}
G.~'t~Hooft, {\em NATO Sci. Ser. B} {\bf 59}, 117  (1980),
  \doi{10.1007/978-1-4684-7571-5_8}.

\bibitem{Egger:2017tkd}
L.~Egger, A.~Maas and R.~Sondenheimer, {\em Mod. Phys. Lett. A} {\bf 32},
  1750212  (2017), \href{http://arxiv.org/abs/1701.02881}{{\ttfamily
  arXiv:1701.02881 [hep-ph]}}, \doi{10.1142/S0217732317502121}.

\bibitem{Maas:2020kda}
A.~Maas and R.~Sondenheimer, {\em Phys. Rev. D} {\bf 102},   113001  (2020),
  \href{http://arxiv.org/abs/2009.06671}{{\ttfamily arXiv:2009.06671
  [hep-ph]}}, \doi{10.1103/PhysRevD.102.113001}.

\bibitem{Elitzur:1975im}
S.~Elitzur, {\em Phys. Rev. D} {\bf 12}, 3978  (1975),
  \doi{10.1103/PhysRevD.12.3978}.

\bibitem{Fradkin:1978dv}
E.~H. Fradkin and S.~H. Shenker, {\em Phys. Rev. D} {\bf 19}, 3682  (1979),
  \doi{10.1103/PhysRevD.19.3682}.

\bibitem{Banks:1979fi}
T.~Banks and E.~Rabinovici, {\em Nucl. Phys. B} {\bf 160}, 349  (1979),
  \doi{10.1016/0550-3213(79)90064-6}.

\bibitem{Gross:1973id}
D.~J. Gross and F.~Wilczek, {\em Phys. Rev. Lett.} {\bf 30}, 1343  (1973),
  \doi{10.1103/PhysRevLett.30.1343}.

\bibitem{Politzer:1973fx}
H.~D. Politzer, {\em Phys. Rev. Lett.} {\bf 30}, 1346  (1973),
  \doi{10.1103/PhysRevLett.30.1346}.

\bibitem{Seiberg:1994pq}
N.~Seiberg, {\em Nucl. Phys. B} {\bf 435}, 129  (1995),
  \href{http://arxiv.org/abs/hep-th/9411149}{{\ttfamily arXiv:hep-th/9411149}},
  \doi{10.1016/0550-3213(94)00023-8}.

\bibitem{Terning:1997xy}
J.~Terning, {\em Phys. Rev. Lett.} {\bf 80}, 2517  (1998),
  \href{http://arxiv.org/abs/hep-th/9706074}{{\ttfamily arXiv:hep-th/9706074}},
  \doi{10.1103/PhysRevLett.80.2517}.

\bibitem{Schmaltz:1998bg}
M.~Schmaltz, {\em Phys. Rev. D} {\bf 59},   105018  (1999),
  \href{http://arxiv.org/abs/hep-th/9805218}{{\ttfamily arXiv:hep-th/9805218}},
  \doi{10.1103/PhysRevD.59.105018}.

\bibitem{Armoni:2008gg}
A.~Armoni, D.~Israel, G.~Moraitis and V.~Niarchos, {\em Phys. Rev. D} {\bf 77},
    105009  (2008), \href{http://arxiv.org/abs/0801.0762}{{\ttfamily
  arXiv:0801.0762 [hep-th]}}, \doi{10.1103/PhysRevD.77.105009}.

\bibitem{Sannino:2009qc}
F.~Sannino, {\em Phys. Rev. D} {\bf 80},   065011  (2009),
  \href{http://arxiv.org/abs/0907.1364}{{\ttfamily arXiv:0907.1364 [hep-th]}},
  \doi{10.1103/PhysRevD.80.065011}.

\bibitem{Sannino:2009me}
F.~Sannino, {\em Nucl. Phys. B} {\bf 830}, 179  (2010),
  \href{http://arxiv.org/abs/0909.4584}{{\ttfamily arXiv:0909.4584 [hep-th]}},
  \doi{10.1016/j.nuclphysb.2009.12.026}.

\bibitem{Mojaza:2011rw}
M.~Mojaza, M.~Nardecchia, C.~Pica and F.~Sannino, {\em Phys. Rev. D} {\bf 83},
   065022  (2011), \href{http://arxiv.org/abs/1101.1522}{{\ttfamily
  arXiv:1101.1522 [hep-th]}}, \doi{10.1103/PhysRevD.83.065022}.

\bibitem{Sannino:2011mr}
F.~Sannino, {\em Mod. Phys. Lett. A} {\bf 26}, 1763  (2011),
  \href{http://arxiv.org/abs/1102.5100}{{\ttfamily arXiv:1102.5100 [hep-ph]}},
  \doi{10.1142/S0217732311036279}.

\bibitem{Karasik:2022gve}
A.~Karasik, K.~{\"O}nder and D.~Tong, {\em JHEP} {\bf 11},   122  (2022),
  \href{http://arxiv.org/abs/2208.07842}{{\ttfamily arXiv:2208.07842
  [hep-th]}}, \doi{10.1007/JHEP11(2022)122}.

\bibitem{Cacciapaglia:2024mfy}
G.~Cacciapaglia and F.~Sannino, {\em Phys. Rev. D} {\bf 111},   035013  (2025),
  \href{http://arxiv.org/abs/2407.17281}{{\ttfamily arXiv:2407.17281
  [hep-ph]}}, \doi{10.1103/PhysRevD.111.035013}.

\bibitem{Gaiotto:2014kfa}
D.~Gaiotto, A.~Kapustin, N.~Seiberg and B.~Willett, {\em JHEP} {\bf 02},   172
  (2015), \href{http://arxiv.org/abs/1412.5148}{{\ttfamily arXiv:1412.5148
  [hep-th]}}, \doi{10.1007/JHEP02(2015)172}.

\bibitem{Cherman:2022eml}
A.~Cherman, T.~Jacobson and M.~Neuzil, {\em JHEP} {\bf 02},   192  (2023),
  \href{http://arxiv.org/abs/2209.00027}{{\ttfamily arXiv:2209.00027
  [hep-th]}}, \doi{10.1007/JHEP02(2023)192}.

\bibitem{Cherman:2023xok}
A.~Cherman and T.~Jacobson, {\em Phys. Rev. D} {\bf 109},   125013  (2024),
  \href{http://arxiv.org/abs/2304.13751}{{\ttfamily arXiv:2304.13751
  [hep-th]}}, \doi{10.1103/PhysRevD.109.125013}.

\bibitem{Nambu:1961tp}
Y.~Nambu and G.~Jona-Lasinio, {\em Phys. Rev.} {\bf 122}, 345  (1961),
  \doi{10.1103/PhysRev.122.345}.

\bibitem{Nambu:1961fr}
Y.~Nambu and G.~Jona-Lasinio, {\em Phys. Rev.} {\bf 124}, 246  (1961),
  \doi{10.1103/PhysRev.124.246}.

\bibitem{Greensite:2011zz}
J.~Greensite, {\em Lect. Notes Phys.} {\bf 821}, 1  (2011),
  \doi{10.1007/978-3-642-14382-3}.

\bibitem{Greensite:2017ajx}
J.~Greensite and K.~Matsuyama, {\em Phys. Rev. D} {\bf 96},   094510  (2017),
  \href{http://arxiv.org/abs/1708.08979}{{\ttfamily arXiv:1708.08979
  [hep-lat]}}, \doi{10.1103/PhysRevD.96.094510}.

\bibitem{Greensite:2018ebg}
J.~Greensite and K.~Matsuyama, {\em PoS} {\bf Confinement2018},   046  (2018),
  \href{http://arxiv.org/abs/1811.01512}{{\ttfamily arXiv:1811.01512
  [hep-lat]}}, \doi{10.22323/1.336.0046}.

\bibitem{Greensite:2023qfx}
J.~Greensite and H.~Y. Yau, {\em Phys. Rev. D} {\bf 109},   034502  (2024),
  \href{http://arxiv.org/abs/2310.01354}{{\ttfamily arXiv:2310.01354
  [hep-lat]}}, \doi{10.1103/PhysRevD.109.034502}.

\bibitem{Dumitrescu:2023hbe}
T.~T. Dumitrescu and P.-S. Hsin, {\em SciPost Phys.} {\bf 17},   093  (2024),
  \href{http://arxiv.org/abs/2312.16898}{{\ttfamily arXiv:2312.16898
  [hep-th]}}, \doi{10.21468/SciPostPhys.17.3.093}.

\bibitem{Intriligator:1995au}
K.~A. Intriligator and N.~Seiberg, {\em Nucl. Phys. B Proc. Suppl.} {\bf 45BC},
  1  (1996), \href{http://arxiv.org/abs/hep-th/9509066}{{\ttfamily
  arXiv:hep-th/9509066}}, \doi{10.1016/0920-5632(95)00626-5}.

\bibitem{Csaki:1996sm}
C.~Csaki, M.~Schmaltz and W.~Skiba, {\em Phys. Rev. Lett.} {\bf 78}, 799
  (1997), \href{http://arxiv.org/abs/hep-th/9610139}{{\ttfamily
  arXiv:hep-th/9610139}}, \doi{10.1103/PhysRevLett.78.799}.

\bibitem{Csaki:1996zb}
C.~Csaki, M.~Schmaltz and W.~Skiba, {\em Phys. Rev. D} {\bf 55}, 7840  (1997),
  \href{http://arxiv.org/abs/hep-th/9612207}{{\ttfamily arXiv:hep-th/9612207}},
  \doi{10.1103/PhysRevD.55.7840}.

\bibitem{Neil:2011yag}
E.~T. Neil, {\em PoS} {\bf LATTICE2011},   009  (2011),
  \href{http://arxiv.org/abs/1205.4706}{{\ttfamily arXiv:1205.4706 [hep-lat]}},
  \doi{10.22323/1.139.0009}.

\bibitem{Gies:2005as}
H.~Gies and J.~Jaeckel, {\em Eur. Phys. J. C} {\bf 46}, 433  (2006),
  \href{http://arxiv.org/abs/hep-ph/0507171}{{\ttfamily arXiv:hep-ph/0507171}},
  \doi{10.1140/epjc/s2006-02475-0}.

\bibitem{Braun:2006jd}
J.~Braun and H.~Gies, {\em JHEP} {\bf 06},   024  (2006),
  \href{http://arxiv.org/abs/hep-ph/0602226}{{\ttfamily arXiv:hep-ph/0602226}},
  \doi{10.1088/1126-6708/2006/06/024}.

\bibitem{Braun:2009ns}
J.~Braun and H.~Gies, {\em JHEP} {\bf 05},   060  (2010),
  \href{http://arxiv.org/abs/0912.4168}{{\ttfamily arXiv:0912.4168 [hep-ph]}},
  \doi{10.1007/JHEP05(2010)060}.

\bibitem{Braun:2010qs}
J.~Braun, C.~S. Fischer and H.~Gies, {\em Phys. Rev. D} {\bf 84},   034045
  (2011), \href{http://arxiv.org/abs/1012.4279}{{\ttfamily arXiv:1012.4279
  [hep-ph]}}, \doi{10.1103/PhysRevD.84.034045}.

\bibitem{Goertz:2024dnz}
F.~Goertz, {\'A}.~Pastor-Guti{\'e}rrez and J.~M. Pawlowski (12 2024),
  \href{http://arxiv.org/abs/2412.12254}{{\ttfamily arXiv:2412.12254
  [hep-th]}}.

\bibitem{Miransky:1998dh}
V.~A. Miransky, {\em Phys. Rev. D} {\bf 59},   105003  (1999),
  \href{http://arxiv.org/abs/hep-ph/9812350}{{\ttfamily arXiv:hep-ph/9812350}},
  \doi{10.1103/PhysRevD.59.105003}.

\bibitem{Miransky:1996pd}
V.~Miransky and K.~Yamawaki, {\em Phys. Rev. D} {\bf 55}, 5051  (1997),
  \href{http://arxiv.org/abs/hep-th/9611142}{{\ttfamily arXiv:hep-th/9611142}},
  \doi{10.1103/PhysRevD.56.3768}, [Erratum: Phys.Rev.D 56, 3768 (1997)].

\bibitem{Kaplan:2009kr}
D.~B. Kaplan, J.-W. Lee, D.~T. Son and M.~A. Stephanov, {\em Phys. Rev. D} {\bf
  80},   125005  (2009), \href{http://arxiv.org/abs/0905.4752}{{\ttfamily
  arXiv:0905.4752 [hep-th]}}, \doi{10.1103/PhysRevD.80.125005}.

\bibitem{Caswell:1974gg}
W.~E. Caswell, {\em Phys. Rev. Lett.} {\bf 33},   244  (1974),
  \doi{10.1103/PhysRevLett.33.244}.

\bibitem{Banks:1981nn}
T.~Banks and A.~Zaks, {\em Nucl. Phys. B} {\bf 196}, 189  (1982),
  \doi{10.1016/0550-3213(82)90035-9}.

\bibitem{Casher:1979vw}
A.~Casher, {\em Phys. Lett. B} {\bf 83}, 395  (1979),
  \doi{10.1016/0370-2693(79)91137-7}.

\bibitem{Banks:1979yr}
T.~Banks and A.~Casher, {\em Nucl. Phys. B} {\bf 169}, 103  (1980),
  \doi{10.1016/0550-3213(80)90255-2}.

\bibitem{Engel:2014cka}
G.~P. Engel, L.~Giusti, S.~Lottini and R.~Sommer, {\em Phys. Rev. Lett.} {\bf
  114},   112001  (2015), \href{http://arxiv.org/abs/1406.4987}{{\ttfamily
  arXiv:1406.4987 [hep-ph]}}, \doi{10.1103/PhysRevLett.114.112001}.

\bibitem{Engel:2014eea}
G.~P. Engel, L.~Giusti, S.~Lottini and R.~Sommer, {\em Phys. Rev. D} {\bf 91},
   054505  (2015), \href{http://arxiv.org/abs/1411.6386}{{\ttfamily
  arXiv:1411.6386 [hep-lat]}}, \doi{10.1103/PhysRevD.91.054505}.

\bibitem{Giusti:2015kwf}
L.~Giusti, {\em PoS} {\bf LATTICE2015},   001  (2016),
  \href{http://arxiv.org/abs/1511.08786}{{\ttfamily arXiv:1511.08786
  [hep-lat]}}, \doi{10.22323/1.251.0001}.

\bibitem{Faber:2017alm}
M.~Faber and R.~H{\"o}llwieser, {\em Prog. Part. Nucl. Phys.} {\bf 97}, 312
  (2017), \href{http://arxiv.org/abs/1908.09740}{{\ttfamily arXiv:1908.09740
  [hep-lat]}}, \doi{10.1016/j.ppnp.2017.08.001}.

\bibitem{Coleman:1980mx}
S.~R. Coleman and E.~Witten, {\em Phys. Rev. Lett.} {\bf 45},   100  (1980),
  \doi{10.1103/PhysRevLett.45.100}.

\bibitem{Veneziano:1980xs}
G.~Veneziano, {\em Phys. Lett. B} {\bf 95}, 90  (1980),
  \doi{10.1016/0370-2693(80)90406-2}.

\bibitem{Sato:2022ayb}
R.~Sato, {\em PTEP} {\bf 2022},   103B03  (2022),
  \href{http://arxiv.org/abs/2202.07664}{{\ttfamily arXiv:2202.07664
  [hep-th]}}, \doi{10.1093/ptep/ptac133}.

\bibitem{Witten:1978bc}
E.~Witten, {\em Nucl. Phys. B} {\bf 149}, 285  (1979),
  \doi{10.1016/0550-3213(79)90243-8}.

\bibitem{Witten:1979vv}
E.~Witten, {\em Nucl. Phys. B} {\bf 156}, 269  (1979),
  \doi{10.1016/0550-3213(79)90031-2}.

\bibitem{Witten:1980sp}
E.~Witten, {\em Annals Phys.} {\bf 128},   363  (1980),
  \doi{10.1016/0003-4916(80)90325-5}.

\bibitem{DiVecchia:1980yfw}
P.~Di~Vecchia and G.~Veneziano, {\em Nucl. Phys. B} {\bf 171}, 253  (1980),
  \doi{10.1016/0550-3213(80)90370-3}.

\bibitem{Witten:1998uka}
E.~Witten, {\em Phys. Rev. Lett.} {\bf 81}, 2862  (1998),
  \href{http://arxiv.org/abs/hep-th/9807109}{{\ttfamily arXiv:hep-th/9807109}},
  \doi{10.1103/PhysRevLett.81.2862}.

\bibitem{Gaiotto:2017yup}
D.~Gaiotto, A.~Kapustin, Z.~Komargodski and N.~Seiberg, {\em JHEP} {\bf 05},
  091  (2017), \href{http://arxiv.org/abs/1703.00501}{{\ttfamily
  arXiv:1703.00501 [hep-th]}}, \doi{10.1007/JHEP05(2017)091}.

\bibitem{witten-jaffe}
A.~Jaffe and E.~Witten, Quantum yang-mills theory, in {\em The Millennium Prize
  Problems\/},  eds. J.~Carlson, A.~Jaffe and W.~Andrew (American Mathematical
  Society and Clay Mathematics Institute, 2006), pp. 129--152.
\newblock \url{https://www.claymath.org/millennium/yang-mills-the-maths-gap/}.

\bibitem{Seiberg:1994bz}
N.~Seiberg, {\em Phys. Rev. D} {\bf 49}, 6857  (1994),
  \href{http://arxiv.org/abs/hep-th/9402044}{{\ttfamily arXiv:hep-th/9402044}},
  \doi{10.1103/PhysRevD.49.6857}.

\bibitem{Aharony:1995zh}
O.~Aharony, J.~Sonnenschein, M.~E. Peskin and S.~Yankielowicz, {\em Phys. Rev.
  D} {\bf 52}, 6157  (1995),
  \href{http://arxiv.org/abs/hep-th/9507013}{{\ttfamily arXiv:hep-th/9507013}},
  \doi{10.1103/PhysRevD.52.6157}.

\bibitem{Cheng:1998xg}
H.-C. Cheng and Y.~Shadmi, {\em Nucl. Phys. B} {\bf 531}, 125  (1998),
  \href{http://arxiv.org/abs/hep-th/9801146}{{\ttfamily arXiv:hep-th/9801146}},
  \doi{10.1016/S0550-3213(98)00539-2}.

\bibitem{Arkani-Hamed:1998dti}
N.~Arkani-Hamed and R.~Rattazzi, {\em Phys. Lett. B} {\bf 454}, 290  (1999),
  \href{http://arxiv.org/abs/hep-th/9804068}{{\ttfamily arXiv:hep-th/9804068}},
  \doi{10.1016/S0370-2693(99)00406-2}.

\bibitem{Luty:1999qc}
M.~A. Luty and R.~Rattazzi, {\em JHEP} {\bf 11},   001  (1999),
  \href{http://arxiv.org/abs/hep-th/9908085}{{\ttfamily arXiv:hep-th/9908085}},
  \doi{10.1088/1126-6708/1999/11/001}.

\bibitem{Abel:2011wv}
S.~Abel, M.~Buican and Z.~Komargodski, {\em Phys. Rev. D} {\bf 84},   045005
  (2011), \href{http://arxiv.org/abs/1105.2885}{{\ttfamily arXiv:1105.2885
  [hep-th]}}, \doi{10.1103/PhysRevD.84.045005}.

\bibitem{Murayama:2021xfj}
H.~Murayama, {\em Phys. Rev. Lett.} {\bf 126},   251601  (2021),
  \href{http://arxiv.org/abs/2104.01179}{{\ttfamily arXiv:2104.01179
  [hep-th]}}, \doi{10.1103/PhysRevLett.126.251601}.

\bibitem{Kondo:2021osz}
D.~Kondo, H.~Murayama, B.~Noether and D.~R. Varier, {\em JHEP} {\bf 04},   152
  (2025), \href{http://arxiv.org/abs/2111.09690}{{\ttfamily arXiv:2111.09690
  [hep-th]}}, \doi{10.1007/JHEP04(2025)152}.

\bibitem{Luzio:2022ccn}
A.~Luzio and L.-X. Xu, {\em JHEP} {\bf 08},   016  (2022),
  \href{http://arxiv.org/abs/2202.01239}{{\ttfamily arXiv:2202.01239
  [hep-th]}}, \doi{10.1007/JHEP08(2022)016}.

\bibitem{Dine:2022req}
M.~Dine and Y.~Yu (4 2022), \href{http://arxiv.org/abs/2205.00115}{{\ttfamily
  arXiv:2205.00115 [hep-ph]}}.

\bibitem{Csaki:2022cyg}
C.~Cs\'aki, A.~Gomes, H.~Murayama, B.~Noether, D.~R. Varier and O.~Telem, {\em
  Phys. Rev. D} {\bf 107},   054015  (2023),
  \href{http://arxiv.org/abs/2212.03260}{{\ttfamily arXiv:2212.03260
  [hep-th]}}, \doi{10.1103/PhysRevD.107.054015}.

\bibitem{deLima:2023ebw}
C.~H. de~Lima and D.~Stolarski, {\em JHEP} {\bf 10},   020  (2023),
  \href{http://arxiv.org/abs/2307.13154}{{\ttfamily arXiv:2307.13154
  [hep-th]}}, \doi{10.1007/JHEP10(2023)020}.

\bibitem{Bai:2025dys}
Y.~Bai, C.~H. de~Lima and D.~Stolarski (6 2025),
  \href{http://arxiv.org/abs/2506.18964}{{\ttfamily arXiv:2506.18964
  [hep-ph]}}.

\bibitem{Vafa:1983tf}
C.~Vafa and E.~Witten, {\em Nucl. Phys. B} {\bf 234}, 173  (1984),
  \doi{10.1016/0550-3213(84)90230-X}.

\bibitem{Kondo:2025njf}
D.~Kondo, H.~Murayama and B.~Noether (5 2025),
  \href{http://arxiv.org/abs/2505.18138}{{\ttfamily arXiv:2505.18138
  [hep-th]}}.

\bibitem{tHooft:1979rat}
G.~'t~Hooft, {\em NATO Sci. Ser. B} {\bf 59}, 135  (1980),
  \doi{10.1007/978-1-4684-7571-5\_9}.

\bibitem{Preskill:1981sr}
J.~Preskill and S.~Weinberg, {\em Phys. Rev. D} {\bf 24},   1059  (1981),
  \doi{10.1103/PhysRevD.24.1059}.

\bibitem{Ciambriello:2022wmh}
L.~Ciambriello, R.~Contino, A.~Luzio, M.~Romano and L.-X. Xu, {\em Phys. Rev.
  D} {\bf 110},   114035  (2024),
  \href{http://arxiv.org/abs/2212.02930}{{\ttfamily arXiv:2212.02930
  [hep-th]}}, \doi{10.1103/PhysRevD.110.114035}.

\bibitem{Ciambriello:2024xzd}
L.~Ciambriello, R.~Contino, A.~Luzio, M.~Romano and L.-X. Xu, {\em Phys. Lett.
  B} {\bf 862},   139367  (2025),
  \href{http://arxiv.org/abs/2404.02967}{{\ttfamily arXiv:2404.02967
  [hep-th]}}, \doi{10.1016/j.physletb.2025.139367}.

\bibitem{Ciambriello:2024msu}
L.~Ciambriello, R.~Contino and L.-X. Xu, {\em Nucl. Phys. B} {\bf 1017},
  116957  (2025), \href{http://arxiv.org/abs/2404.02971}{{\ttfamily
  arXiv:2404.02971 [hep-th]}}, \doi{10.1016/j.nuclphysb.2025.116957}.

\bibitem{Kosower:1984aw}
D.~A. Kosower, {\em Phys. Lett. B} {\bf 144}, 215  (1984),
  \doi{10.1016/0370-2693(84)91806-9}.

\bibitem{Bars:1981se}
I.~Bars and S.~Yankielowicz, {\em Phys. Lett. B} {\bf 101}, 159  (1981),
  \doi{10.1016/0370-2693(81)90664-X}.

\bibitem{Poppitz:2019fnp}
E.~Poppitz and T.~A. Ryttov, {\em Phys. Rev. D} {\bf 100},   091901  (2019),
  \href{http://arxiv.org/abs/1904.11640}{{\ttfamily arXiv:1904.11640
  [hep-th]}}, \doi{10.1103/PhysRevD.100.091901}.

\bibitem{Dimopoulos:1981xc}
S.~Dimopoulos and J.~Preskill, {\em Nucl. Phys. B} {\bf 199}, 206  (1982),
  \doi{10.1016/0550-3213(82)90345-5}.

\bibitem{Valenti:2023olg}
A.~Valenti and L.~Vecchi, {\em JHEP} {\bf 02},   052  (2024),
  \href{http://arxiv.org/abs/2306.13088}{{\ttfamily arXiv:2306.13088
  [hep-ph]}}, \doi{10.1007/JHEP02(2024)052}.

\bibitem{Frishman:1980dq}
Y.~Frishman, A.~Schwimmer, T.~Banks and S.~Yankielowicz, {\em Nucl. Phys. B}
  {\bf 177}, 157  (1981), \doi{10.1016/0550-3213(81)90268-6}.

\bibitem{Farrar:1980sn}
G.~R. Farrar, {\em Phys. Lett. B} {\bf 96}, 273  (1980),
  \doi{10.1016/0370-2693(80)90765-0}.

\bibitem{Schwimmer:1981yy}
A.~Schwimmer, {\em Nucl. Phys. B} {\bf 198}, 269  (1982),
  \doi{10.1016/0550-3213(82)90557-0}.

\bibitem{Cohen:1981iz}
E.~Cohen and Y.~Frishman, {\em Phys. Lett. B} {\bf 109}, 35  (1982),
  \doi{10.1016/0370-2693(82)90457-9}.

\bibitem{Kaul:1981fd}
R.~K. Kaul and R.~Rajaraman, {\em Phys. Lett. B} {\bf 110}, 385  (1982),
  \doi{10.1016/0370-2693(82)91278-3}.

\bibitem{Takeshita:1981sx}
S.~Takeshita, H.~Komatsu, A.~Kakuto and K.~Inoue, {\em Prog. Theor. Phys.} {\bf
  66},   2221  (1981), \doi{10.1143/PTP.66.2221}.

\bibitem{Adler:1969er}
S.~L. Adler and W.~A. Bardeen, {\em Phys. Rev.} {\bf 182}, 1517  (1969),
  \doi{10.1103/PhysRev.182.1517}.

\bibitem{Weinberg:1996kr}
S.~Weinberg, {\em {The quantum theory of fields. Vol. 2: Modern applications}}
  (Cambridge University Press, 8 2013).

\bibitem{Alvarez-Gaume:1984zlq}
L.~Alvarez-Gaume and P.~H. Ginsparg, {\em Annals Phys.} {\bf 161},   423
  (1985), \doi{10.1016/0003-4916(85)90087-9}, [Erratum: Annals Phys. 171, 233
  (1986)].

\bibitem{Alvarez-Gaume:1985zzv}
L.~Alvarez-Gaume, {\em NATO Sci. Ser. B} {\bf 141}  (1986),
  \doi{10.1007/978-1-4757-0363-4_4}.

\bibitem{Harvey:2005it}
J.~A. Harvey, { {TASI 2003 lectures on anomalies}} (9 2005),
  \href{http://arxiv.org/abs/hep-th/0509097}{{\ttfamily arXiv:hep-th/0509097}}.

\bibitem{Bilal:2008qx}
A.~Bilal (2 2008), \href{http://arxiv.org/abs/0802.0634}{{\ttfamily
  arXiv:0802.0634 [hep-th]}}.

\bibitem{Witten:2015aba}
E.~Witten, {\em Rev. Mod. Phys.} {\bf 88},   035001  (2016),
  \href{http://arxiv.org/abs/1508.04715}{{\ttfamily arXiv:1508.04715
  [cond-mat.mes-hall]}}, \doi{10.1103/RevModPhys.88.035001}.

\bibitem{Tachikawa_TASI}
Y.~Tachikawa, in {\em Lecture on anomalies and topological phases (2019)\/},
\newblock \url{https://member.ipmu.jp/yuji.tachikawa/lectures/2019-top-anom/}.

\bibitem{Arouca:2022psl}
R.~Arouca, A.~Cappelli and T.~H. Hansson, {\em SciPost Phys. Lect. Notes} {\bf
  62},  ~1  (2022), \href{http://arxiv.org/abs/2204.02158}{{\ttfamily
  arXiv:2204.02158 [cond-mat.str-el]}}, \doi{10.21468/SciPostPhysLectNotes.62}.

\bibitem{tong_gauge_theory}
D.~Tong, in {\em Lectures on Gauge Theory\/},
\newblock \url{https://www.damtp.cam.ac.uk/user/tong/gaugetheory.html}.

\bibitem{Cheng:2022sgb}
M.~Cheng and N.~Seiberg, {\em SciPost Phys.} {\bf 15},   051  (2023),
  \href{http://arxiv.org/abs/2211.12543}{{\ttfamily arXiv:2211.12543
  [cond-mat.str-el]}}, \doi{10.21468/SciPostPhys.15.2.051}.

\bibitem{Ye:2023uoz}
W.~Ye, {Aspects of Anomaly in Condensed Matter Physics}, PhD thesis, U.
  Waterloo (main)  (2023).

\bibitem{Cordova:2019jnf}
C.~C{\'o}rdova, D.~S. Freed, H.~T. Lam and N.~Seiberg, {\em SciPost Phys.} {\bf
  8},   001  (2020), \href{http://arxiv.org/abs/1905.09315}{{\ttfamily
  arXiv:1905.09315 [hep-th]}}, \doi{10.21468/SciPostPhys.8.1.001}.

\bibitem{Cordova:2019uob}
C.~C{\'o}rdova, D.~S. Freed, H.~T. Lam and N.~Seiberg, {\em SciPost Phys.} {\bf
  8},   002  (2020), \href{http://arxiv.org/abs/1905.13361}{{\ttfamily
  arXiv:1905.13361 [hep-th]}}, \doi{10.21468/SciPostPhys.8.1.002}.

\bibitem{Freed:2004yc}
D.~S. Freed and G.~W. Moore, {\em Commun. Math. Phys.} {\bf 263}, 89  (2006),
  \href{http://arxiv.org/abs/hep-th/0409135}{{\ttfamily arXiv:hep-th/0409135}},
  \doi{10.1007/s00220-005-1482-7}.

\bibitem{Freed:2019jzd}
D.~S. Freed and M.~J. Hopkins, {\em Adv. Theor. Math. Phys.} {\bf 24}, 1773
  (2020), \href{http://arxiv.org/abs/1901.06419}{{\ttfamily arXiv:1901.06419
  [math-ph]}}, \doi{10.4310/ATMP.2020.v24.n7.a3}.

\bibitem{Freed:2014iua}
D.~S. Freed, {\em Proc. Symp. Pure Math.} {\bf 88}, 25  (2014),
  \href{http://arxiv.org/abs/1404.7224}{{\ttfamily arXiv:1404.7224 [hep-th]}},
  \doi{10.1090/pspum/088/01462}.

\bibitem{Callan:1984sa}
C.~G. Callan, Jr. and J.~A. Harvey, {\em Nucl. Phys. B} {\bf 250}, 427  (1985),
  \doi{10.1016/0550-3213(85)90489-4}.

\bibitem{Chen:2011pg}
X.~Chen, Z.-C. Gu, Z.-X. Liu and X.-G. Wen, {\em Phys. Rev. B} {\bf 87},
  155114  (2013), \href{http://arxiv.org/abs/1106.4772}{{\ttfamily
  arXiv:1106.4772 [cond-mat.str-el]}}, \doi{10.1103/PhysRevB.87.155114}.

\bibitem{Senthil:2014ooa}
T.~Senthil, {\em Ann. Rev. Condensed Matter Phys.} {\bf 6},   299  (2015),
  \href{http://arxiv.org/abs/1405.4015}{{\ttfamily arXiv:1405.4015
  [cond-mat.str-el]}}, \doi{10.1146/annurev-conmatphys-031214-014740}.

\bibitem{Lieb:1961fr}
E.~H. Lieb, T.~Schultz and D.~Mattis, {\em Annals Phys.} {\bf 16}, 407  (1961),
  \doi{10.1016/0003-4916(61)90115-4}.

\bibitem{Affleck:1986pq}
I.~Affleck and E.~H. Lieb, {\em Lett. Math. Phys.} {\bf 12},  ~57  (1986),
  \doi{10.1007/BF00400304}.

\bibitem{Oshikawa:2000lrt}
M.~Oshikawa, {\em Phys. Rev. Lett.} {\bf 84},   3370  (2000),
  \href{http://arxiv.org/abs/cond-mat/0002392}{{\ttfamily
  arXiv:cond-mat/0002392}}, \doi{10.1103/PhysRevLett.84.3370}.

\bibitem{Hastings:2003zx}
M.~B. Hastings, {\em Phys. Rev. B} {\bf 69},   104431  (2004),
  \href{http://arxiv.org/abs/cond-mat/0305505}{{\ttfamily
  arXiv:cond-mat/0305505}}, \doi{10.1103/PhysRevB.69.104431}.

\bibitem{Affleck:1988nt}
I.~Affleck, {\em J. Phys. C} {\bf 1},   3047  (1989),
  \doi{10.1088/0953-8984/1/19/001}.

\bibitem{Tasaki:2022gka}
H.~Tasaki (2 2022), \href{http://arxiv.org/abs/2202.06243}{{\ttfamily
  arXiv:2202.06243 [cond-mat.stat-mech]}}.

\bibitem{Jensen:2017eof}
K.~Jensen, E.~Shaverin and A.~Yarom, {\em JHEP} {\bf 01},   085  (2018),
  \href{http://arxiv.org/abs/1710.07299}{{\ttfamily arXiv:1710.07299
  [hep-th]}}, \doi{10.1007/JHEP01(2018)085}.

\bibitem{Thorngren:2020yht}
R.~Thorngren and Y.~Wang, {\em JHEP} {\bf 09},   017  (2021),
  \href{http://arxiv.org/abs/2012.15861}{{\ttfamily arXiv:2012.15861
  [hep-th]}}, \doi{10.1007/JHEP09(2021)017}.

\bibitem{Li:2022drc}
L.~Li, C.-T. Hsieh, Y.~Yao and M.~Oshikawa, {\em Phys. Rev. B} {\bf 110},
  045118  (2024), \href{http://arxiv.org/abs/2205.11190}{{\ttfamily
  arXiv:2205.11190 [hep-th]}}, \doi{10.1103/PhysRevB.110.045118}.

\bibitem{Chen:2023hmm}
J.-W. Chen, C.-T. Hsieh and R.~Matsudo, {\em SciPost Phys.} {\bf 17},   068
  (2024), \href{http://arxiv.org/abs/2306.10845}{{\ttfamily arXiv:2306.10845
  [hep-th]}}, \doi{10.21468/SciPostPhys.17.2.068}.

\bibitem{Choi:2023xjw}
Y.~Choi, B.~C. Rayhaun, Y.~Sanghavi and S.-H. Shao, {\em Phys. Rev. D} {\bf
  108},   125005  (2023), \href{http://arxiv.org/abs/2305.09713}{{\ttfamily
  arXiv:2305.09713 [hep-th]}}, \doi{10.1103/PhysRevD.108.125005}.

\bibitem{Seifnashri:2023dpa}
S.~Seifnashri, {\em SciPost Phys.} {\bf 16},   098  (2024),
  \href{http://arxiv.org/abs/2308.05151}{{\ttfamily arXiv:2308.05151
  [cond-mat.str-el]}}, \doi{10.21468/SciPostPhys.16.4.098}.

\bibitem{Yonekura:2019vyz}
K.~Yonekura, {\em JHEP} {\bf 05},   062  (2019),
  \href{http://arxiv.org/abs/1901.08188}{{\ttfamily arXiv:1901.08188
  [hep-th]}}, \doi{10.1007/JHEP05(2019)062}.

\bibitem{Tanizaki:2018wtg}
Y.~Tanizaki, {\em JHEP} {\bf 08},   171  (2018),
  \href{http://arxiv.org/abs/1807.07666}{{\ttfamily arXiv:1807.07666
  [hep-th]}}, \doi{10.1007/JHEP08(2018)171}.

\bibitem{Morikawa:2022liz}
O.~Morikawa, H.~Wada and S.~Yamaguchi, {\em Phys. Rev. D} {\bf 107},   045020
  (2023), \href{http://arxiv.org/abs/2211.12079}{{\ttfamily arXiv:2211.12079
  [hep-th]}}, \doi{10.1103/PhysRevD.107.045020}.

\bibitem{Coleman:1982yg}
S.~R. Coleman and B.~Grossman, {\em Nucl. Phys. B} {\bf 203}, 205  (1982),
  \doi{10.1016/0550-3213(82)90028-1}.

\bibitem{Wess:1971yu}
J.~Wess and B.~Zumino, {\em Phys. Lett. B} {\bf 37}, 95  (1971),
  \doi{10.1016/0370-2693(71)90582-X}.

\bibitem{Witten:1983tw}
E.~Witten, {\em Nucl. Phys. B} {\bf 223}, 422  (1983),
  \doi{10.1016/0550-3213(83)90063-9}.

\bibitem{Lee:2020ojw}
Y.~Lee, K.~Ohmori and Y.~Tachikawa, {\em SciPost Phys.} {\bf 10},   061
  (2021), \href{http://arxiv.org/abs/2009.00033}{{\ttfamily arXiv:2009.00033
  [hep-th]}}, \doi{10.21468/SciPostPhys.10.3.061}.

\bibitem{Yonekura:2020upo}
K.~Yonekura, {\em JHEP} {\bf 03},   057  (2021),
  \href{http://arxiv.org/abs/2009.04692}{{\ttfamily arXiv:2009.04692
  [hep-th]}}, \doi{10.1007/JHEP03(2021)057}.

\bibitem{Weinberg:1980kq}
S.~Weinberg and E.~Witten, {\em Phys. Lett. B} {\bf 96}, 59  (1980),
  \doi{10.1016/0370-2693(80)90212-9}.

\bibitem{Stora:1983ct}
R.~Stora, { {ALGEBRAIC STRUCTURE AND TOPOLOGICAL ORIGIN OF ANOMALIES}} (11
  1983).

\bibitem{Zumino:1983ew}
B.~Zumino, { {CHIRAL ANOMALIES AND DIFFERENTIAL GEOMETRY: LECTURES GIVEN AT LES
  HOUCHES, AUGUST 1983}}, in {\em {Les Houches Summer School on Theoretical
  Physics: Relativity, Groups and Topology}\/},  (10 1983), pp. 1291--1322.

\bibitem{Manes:1985df}
J.~Manes, R.~Stora and B.~Zumino, {\em Commun. Math. Phys.} {\bf 102},   157
  (1985), \doi{10.1007/BF01208825}.

\bibitem{Cordova:2018cvg}
C.~C{\'o}rdova, T.~T. Dumitrescu and K.~Intriligator, {\em JHEP} {\bf 02},
  184  (2019), \href{http://arxiv.org/abs/1802.04790}{{\ttfamily
  arXiv:1802.04790 [hep-th]}}, \doi{10.1007/JHEP02(2019)184}.

\bibitem{Yamatsu:2015npn}
N.~Yamatsu (11 2015), \href{http://arxiv.org/abs/1511.08771}{{\ttfamily
  arXiv:1511.08771 [hep-ph]}}.

\bibitem{Vafa:1984xg}
C.~Vafa and E.~Witten, {\em Phys. Rev. Lett.} {\bf 53},   535  (1984),
  \doi{10.1103/PhysRevLett.53.535}.

\bibitem{Weingarten:1983uj}
D.~Weingarten, {\em Phys. Rev. Lett.} {\bf 51},   1830  (1983),
  \doi{10.1103/PhysRevLett.51.1830}.

\bibitem{Nussinov:1999sx}
S.~Nussinov and M.~A. Lampert, {\em Phys. Rept.} {\bf 362}, 193  (2002),
  \href{http://arxiv.org/abs/hep-ph/9911532}{{\ttfamily arXiv:hep-ph/9911532}},
  \doi{10.1016/S0370-1573(01)00091-6}.

\bibitem{Kogan:1998zc}
I.~I. Kogan, A.~Kovner and M.~A. Shifman, {\em Phys. Rev. D} {\bf 59},   016001
   (1999), \href{http://arxiv.org/abs/hep-ph/9807286}{{\ttfamily
  arXiv:hep-ph/9807286}}, \doi{10.1103/PhysRevD.59.016001}.

\bibitem{Gagliano:2025gwr}
F.~Gagliano, A.~Grigoletto and K.~Ohmori (1 2025),
  \href{http://arxiv.org/abs/2501.09069}{{\ttfamily arXiv:2501.09069
  [hep-th]}}.

\bibitem{Wilson:1974sk}
K.~G. Wilson, {\em Phys. Rev. D} {\bf 10}, 2445  (1974),
  \doi{10.1103/PhysRevD.10.2445}.

\bibitem{Kogut:1974ag}
J.~B. Kogut and L.~Susskind, {\em Phys. Rev. D} {\bf 11}, 395  (1975),
  \doi{10.1103/PhysRevD.11.395}.

\bibitem{Susskind:1976jm}
L.~Susskind, {\em Phys. Rev. D} {\bf 16}, 3031  (1977),
  \doi{10.1103/PhysRevD.16.3031}.

\bibitem{Nielsen:1980rz}
H.~B. Nielsen and M.~Ninomiya, {\em Nucl. Phys. B} {\bf 185},  ~20  (1981),
  \doi{10.1016/0550-3213(82)90011-6}, [Erratum: Nucl.Phys.B 195, 541 (1982)].

\bibitem{Nielsen:1981xu}
H.~B. Nielsen and M.~Ninomiya, {\em Nucl. Phys. B} {\bf 193}, 173  (1981),
  \doi{10.1016/0550-3213(81)90524-1}.

\bibitem{Nielsen:1981hk}
H.~B. Nielsen and M.~Ninomiya, {\em Phys. Lett. B} {\bf 105}, 219  (1981),
  \doi{10.1016/0370-2693(81)91026-1}.

\bibitem{Chandrasekharan:2004dph}
S.~Chandrasekharan and U.~J. Wiese, {\em Prog. Part. Nucl. Phys.} {\bf 53}, 373
   (2004), \href{http://arxiv.org/abs/hep-lat/0405024}{{\ttfamily
  arXiv:hep-lat/0405024}}, \doi{10.1016/j.ppnp.2004.05.003}.

\bibitem{Kaplan:2009yg}
D.~B. Kaplan, { {Chiral Symmetry and Lattice Fermions}}, in {\em {Les Houches
  Summer School: Session 93: Modern perspectives in lattice QCD: Quantum field
  theory and high performance computing}\/},  (12 2009), pp. 223--272,
  \href{http://arxiv.org/abs/0912.2560}{{\ttfamily arXiv:0912.2560 [hep-lat]}}.

\bibitem{Creutz:2011ati}
M.~Creutz, {\em Acta Phys. Slov.} {\bf 61}, 1  (2011),
  \href{http://arxiv.org/abs/1103.3304}{{\ttfamily arXiv:1103.3304 [hep-lat]}},
  \doi{10.2478/v10155-011-0001-y}.

\bibitem{Wang:2018jkc}
J.~Wang and X.-G. Wen, {\em Phys. Rev. Res.} {\bf 2},   023356  (2020),
  \href{http://arxiv.org/abs/1809.11171}{{\ttfamily arXiv:1809.11171
  [hep-th]}}, \doi{10.1103/PhysRevResearch.2.023356}.

\bibitem{Li:2025tvu}
H.-L. Li, {\'A}.~Pastor-Guti{\'e}rrez, S.~Vatani and L.-X. Xu (7 2025),
  \href{http://arxiv.org/abs/2507.21208}{{\ttfamily arXiv:2507.21208
  [hep-th]}}.

\bibitem{Bolognesi:2021jzs}
S.~Bolognesi, K.~Konishi and A.~Luzio, {\em Int. J. Mod. Phys. A} {\bf 37},
  2230014  (2022), \href{http://arxiv.org/abs/2110.02104}{{\ttfamily
  arXiv:2110.02104 [hep-th]}}, \doi{10.1142/S0217751X22300149}.

\bibitem{Bolognesi:2023sxe}
S.~Bolognesi, K.~Konishi and A.~Luzio, {\em J. Phys. Conf. Ser.} {\bf 2531},
  012006  (2023), \href{http://arxiv.org/abs/2304.03357}{{\ttfamily
  arXiv:2304.03357 [hep-th]}}, \doi{10.1088/1742-6596/2531/1/012006}.

\bibitem{Konishi:2024rjz}
K.~Konishi, S.~Bolognesi and A.~Luzio, {\em PoS} {\bf CORFU2023},   126
  (2024), \href{http://arxiv.org/abs/2403.15775}{{\ttfamily arXiv:2403.15775
  [hep-th]}}, \doi{10.22323/1.463.0126}.

\bibitem{BahaBalantekin:1981fmm}
A.~Baha~Balantekin and I.~Bars, {\em J. Math. Phys.} {\bf 23},   1239  (1982),
  \doi{10.1063/1.525508}.

\bibitem{Bars:1982ps}
I.~Bars, {\em Lectures Appl. Math.} {\bf 21},  ~17  (1983).

\bibitem{Bars:1984rb}
I.~Bars, {\em Physica D} {\bf 15},  ~42  (1985),
  \doi{10.1016/0167-2789(85)90147-2}.

\end{thebibliography}

\end{document}